\documentclass{statsoc}

\usepackage{amsmath,amssymb,fullpage,graphicx,color,soul,pict2e,curve2e, bbm, xcolor}
\usepackage[normalem]{ulem}

\title[Short title]{Sensitivity Analysis for Marginal Structural Models}
\author[Matteo Bonvini {\it et al.}]{Matteo Bonvini, Edward H. Kennedy, Val\'erie Ventura and Larry Wasserman}
\address{Carnegie Mellon University, Pittsburgh, USA.}

\usepackage{comment}
\usepackage{subcaption}
\let\hat\widehat
\let\tilde\widetilde

\newtheorem{theorem}{Theorem}
\newtheorem{lemma}[theorem]{Lemma}

\newtheorem{proposition}[theorem]{Proposition}

\DeclareMathOperator*{\argmin}{argmin}
\renewcommand{\P}{\mbox{$\mathbb{P}$}}
\newcommand{\E}{\mbox{$\mathbb{E}$}}
\newcommand{\one}{\mathbbm{1}}
\usepackage{hyperref}
\hypersetup{
	colorlinks   = true,
	linkcolor = red,
	citecolor    = blue
}

\def\inprob{\stackrel{p}{\rightarrow}}
\def\indist{\rightsquigarrow}
\def\ind{\perp\!\!\!\perp}

\usepackage{amssymb,mathtools}
\DeclarePairedDelimiter{\Norm}{\lVert}{\rVert}
\newcommand{\var}{\text{var}}

\newcommand{\Pb}{\mathbb{P}}
\newcommand{\Ub}{\mathbb{U}}

\newcommand{\Pn}{\mathbb{P}_n}

\newcommand{\R}{\mathbb{R}}

\newcommand{\pihat}{\widehat\pi}
\newcommand{\muhat}{\widehat\mu}
\newcommand{\betahat}{\widehat\beta}

\newcommand{\phat}{\widehat{p}}
\newcommand{\mhat}{\widehat{m}}

\newcommand{\rhat}{\widehat{r}}

\newcommand{\Rhat}{\widehat{R}}

\newcommand{\vhat}{\widehat{v}}
\newcommand{\qhat}{\widehat{q}}

\newcommand{\Ehat}{\widehat{\E}}

\DeclareMathOperator{\sgn}{sgn}

\definecolor{myblue}{RGB}{50,50,150}
\def\VV{\textcolor{magenta}}

\DeclareRobustCommand{\Gali}{%
  \begingroup\setlength{\unitlength}{\fontcharht\font`A}%
\unitlength =0.8mm\relax
\linethickness{0.6pt}
\begin{picture}(4.2,5)
\hspace{0.03in} \Curve(1,1.9)<-0.3,1>(0.5,3.5)<-0.05,2>[-0.05,-2](-0.5,3.5)<-0.1,2>[-0.1,-2](-1.8,1)<-0.1,-2>[0.1,-2](-0.5,-0.5)<1,0>(1,0.5)<-0.1,1>[0.1,-0.1](-0.2,-1.0)<-2,-0.5>(-0.3,-2)<2,-1> 
\end{picture}
  \endgroup
}

\begin{document}

\begin{abstract}
We introduce several methods
for assessing sensitivity to unmeasured
confounding 
in marginal structural models; importantly we allow
treatments to  be discrete or continuous, 
static or time-varying.
We consider three sensitivity models: a propensity-based model, an outcome-based model,
and a subset confounding model, in which only a fraction of the population is
subject to unmeasured confounding.
In each case we develop efficient estimators and confidence intervals 
for bounds on the causal parameters. 
\end{abstract}


\section{Introduction}

Marginal structural models (MSMs)
\citep{robins1998marginal, robins2000marginal, robins2000marginal2}
are a class of semiparametric model
commonly used for causal
inference.
As is typical in causal inference,
the parameters of the model are only identified
under an assumption of no unmeasured confounding.
Thus, it is important to quantify
how sensitive the inferences are to
this assumption.
Most existing sensitivity analysis methods deal with binary
point treatments. In contrast, in this paper we develop
tools for assessing sensitivity for MSMs with both continuous (non-binary) and time-varying treatments.

For simplicity, consider the static treatment setting first. Extensions 
to time-varying treatments are described in Section~\ref{section::time}.
Suppose we have $n$ iid observations $(Z_1, \ldots, Z_n)$, 
with $Z_i = (X_i,A_i,Y_i)$ from a distribution $\Pb$,
where $Y\in\mathbb{R}$ is the outcome of interest,
$A\in\mathbb{R}$ is a treatment (or exposure)
and $X\in\mathbb{R}^d$ is a vector
of confounding variables.
Define the collection of counterfactual random variables
(also called potential outcomes)
$\{ Y(a):\ a\in\mathbb{R} \}$,
where $Y(a)$ denotes the value that $Y$ would have
if $A$ were set to $a$.
The usual assumptions in causal inference are:
\begin{itemize}
\item[(A1)] No interference: if $A=a$ then $Y=Y(a)$, meaning that a subject's potential outcomes only
depend on their own treatment. 
\item[(A2)] Overlap: $\pi(a|x)>0$ for all $x$ and $a$,
where $\pi(a|x)$ is the density of $A$ given $X=x$ (the {\em propensity score}). Overlap guarantees that all subjects
have some chance of receiving each treatment level.
\item[(A3)] No unmeasured confounding:
the counterfactuals $\{ Y(a):\ a\in\mathbb{R} \}$
are independent of $A$ given the observed covariates $X$. This assumption means that
the treatment is as good as randomized
within levels of the measured covariates; in other words, there are no
unmeasured variables $U$ that affect both $A$ and~$Y$.
\end{itemize}
Under these assumptions, 
the causal mean $\E \{ Y(a) \}$ is identified and
equal to 
\begin{equation}\label{eq::psi}
\psi(a)  \equiv \int \mu(x,a) d\Pb(x),
\end{equation}
where $\mu(x,a) = \E[Y |  X=x,A=a]$ is the outcome regression (causal parameters other than $\E \{ Y(a) \}$, e.g., cumulative distribution functions, are identified similarly).
Equation (\ref{eq::psi}) is a special case of the $g$-formula \citep{robins1986new}.

A marginal structural model (MSM)
is a semiparametric model assuming $\psi(a) = g(a;\beta)$ 
\citep{robins1998marginal, robins2000marginal, robins2000marginal2}. 
The MSM provides an interpretable model for the treatment effect
and $\beta$ 
can be estimated
using simple estimating equations.
The model is semiparametric in the sense that it leaves
the data generating distribution unspecified except
for the restriction that
$\int \mu(x,a) d\Pb(x) = g(a;\beta)$.
If $g$ is mis-specified, one can regard $g(a;\beta)$ as an approximation to $\psi(a)$,
in which case one estimates
the value $\beta_*$ that minimizes
$\int (\psi(a)-g(a;\beta))^2 \omega(a) da$,
where $\omega$ is a user provided weight
function \citep{neugebauer2007nonparametric}.

In practice, there are often unmeasured
confounders $U$ so that assumption (A3)
fails. This is especially true for observational studies where
treatment is not under investigators' control, but it can also occur in
experiments in the presence of non-compliance. 
In these cases, $\E \{ Y(a) \}$ is no longer identified.
We can still estimate the functional
$\psi(a)$ in (\ref{eq::psi}) but we no longer have 
$\E \{ Y(a) \} = \psi(a)$.
Sensitivity analysis methods aim to assess how
much $\E \{ Y(a) \}$ and the MSM parameter $\beta$ will change when such unmeasured confounders
$U$ exist.  In this paper, we will derive bounds for $\E \{ Y(a) \} \equiv g(a; \beta)$, as well as for 
$\beta$, under varying amounts of unmeasured confounding.

We consider several sensitivity models for
unmeasured confounding:
a propensity-based model, an outcome-based model,
and a subset confounding model, in which only a fraction of the population is
subject to unmeasured confounding.

\subsection{Related Work}

Sensitivity analysis for causal inference
began with
\cite{cornfield}.
Theory and methods for
sensitivity analysis 
were greatly expanded by
\cite{Rosenbaum}.
Recently, there has been a flurry of interest in sensitivity analysis
including 
\citet{chernozhukov2021omitted, kallus2019interval,
zhao2017sensitivity,
yadlowsky2018bounds,ed}, among others.
We refer to Section 2 of \citet{ed} for a review. 
Most work deals with binary, static treatments. 

The closest work to ours is
\cite{brumback2004sensitivity},
who study sensitivity for
MSMs with binary treatments
using parametric models
for the sensitivity analysis.
We instead consider nonparametric sensitivity models, for continuous rather than binary treatments.
While completing this paper,
\cite{dorn2021sharp} appeared on arXiv,
who independently derived bounds on treatment effects for 
nonparametric causal models that are similar to our bounds in 
Section~\ref{sec::prelim}, Lemma~\ref{lemma:np_cond_bounds}.
Here we treat MSMs rather than nonparametric causal models, with
Lemma~\ref{lemma:np_cond_bounds} being an intermediate step 
to our results.

\subsection{Outline}

We first treat the static treatment setting.
In Section~\ref{section::MSM}
we review MSMs.
In Section~\ref{section::sens}
we introduce our three sensitivity analysis models.
We find bounds for the MSM $g(a; \beta)$ and for its parameter $\beta$
under propensity sensitivity in Section~\ref{section::MSMPS}, under outcome sensitivity in
Section~\ref{section::MSMOS} and under subset sensitivity in
Appendix~\ref{section::subset2}.
Then in Section~\ref{section::time}, we extend our methods to the time series setting.
We illustrate
our methods on simulated data in Appendix~\ref{app:examples} and on observational data in Section~\ref{section::examples}. Section~\ref{section::conclusion} contains concluding remarks.
All proofs can be found in the Appendix.

\subsection{Notation}

We use the notation $\Pb[ f(Z)] = \int f(z) d\Pb(z)$ and
$\mathbb{U} [f(Z_1, Z_2)] = \int f(z_1, z_2)d\Pb(z_1, z_2)$ to denote
expectations of a fixed function, and $\Pn [f(Z)] =
n^{-1}\sum_{i = 1}^n f(Z_i)$ and $\mathbb{U}_n [f(Z_1, Z_2)] =
\{n(n-1)\}^{-1} \sum_{1 \leq i \neq j \leq n}^n f(Z_i, Z_j)$ to denote
their sample counterparts, where $\mathbb{U}_n$ is the
usual $U$-statistic measure. We also let $\|f\|^2 = \int f^2(z)
d\Pb(z)$ denote the $L^2(\Pb)$ norm of $f$ and $\|f\|_\infty =
\sup_z|f(z)|$ denote the $L^\infty$ or sup-norm of $f$. 
For $\beta \in \R^k$
we let $\| \beta \|$
denote the Euclidean norm.
For $f(z_1,z_2)$ we let
$S_2 [f] = \{f(z_1, z_2) + f(z_2, z_1)\} / 2$ be the symmetrizing function.
Then $\mathbb{U}_n [f(Z_1, Z_2)] = \mathbb{U}_n [S_2 [f(Z_1, Z_2)]]$.

\subsection{Some Inferential Issues \label{sec::issues}}

Here we briefly discuss
three issues that commonly arise in this paper when constructing confidence intervals.

The first is that we often have to
estimate quantities
of the form
$\nu = \int\int f(x,a)\pi(a) da  d\Pb(x)$
where $\pi(a)$ is the marginal density of $A$.
This is not a usual expected value since the integral is with respect to a product of marginals, 
$\pi(a) d\Pb(x)$, rather than the joint measure $\Pb(x,a)$.
Then $\nu$ can be written as
$$
\Ub[ f(Z_1,Z_2)] \equiv \int\int \frac{1}{2} \left[f(x_1,a_2)  + f(x_2,a_1)\right] d\P(x_1,a_1)d\P(x_2,a_2)
= \int\int g(z_1,z_2) d\P(z_1) d\P(z_2)
$$
where $Z_1 =(X_1,A_1,Y_1)$ and $Z_2 =(X_2,A_2,Y_2)$ 
are two independent draws and
$g(z_1,z_2) = S_2[ f] \equiv (f(x_1,a_2) + f(x_2,a_1))/2$. 
Under certain conditions,
the limiting distribution of 
$\sqrt{n}\{ \mathbb{U}_n [\widehat{f}(Z_1,Z_2)] - \Ub [f(Z_1, Z_2)]\}$, where $\hat f$ is an estimate of $f$, is the same as that of
$\sqrt{n}(\mathbb{U}_n - \mathbb{U}) [f(Z_1, Z_2)]$. 
More specifically, let $\alpha \in
\R^k$, where $k$ is the dimension of $f$. By Theorem 12.3 in
\cite{van2000asymptotic},
$$
\sqrt{n}(\mathbb{U}_n - \mathbb{U}) [\alpha^T f(Z_1, Z_2)] \to N(0, 4\sigma^2),
$$ 
where 
$\sigma^2  =  \frac{1}{4}\alpha^T\Sigma\alpha$
and
$\Sigma = \var\left[ \int S_2 [f(Z_1, z_2) ]d\Pb(z_2) \right]$.
Therefore, by the Cramer-Wold device, $\sqrt{n}(\mathbb{U}_n -
\mathbb{U}) [f(Z_1, Z_2)]  \indist N(0, \Sigma)$. 
Thus,
$\sqrt{n}(\mathbb{U}_n - \mathbb{U}) [S_2 [f(Z_1, Z_2)]]$ has variance
equal to the variance of the influence function of $\nu = \int\int f(x,a)\pi(a)  da d\Pb(x)$ and thus it
is  efficient.

The second issue is that
calculating the variances of these estimators can be cumbersome. Instead,
we construct confidence intervals using the HulC \citep{kuchibhotla2021hulc},
which avoids estimating variances. The dataset is
randomly split into $B = \log(2/\alpha)/\log 2$ subsamples ($B=6$ when $\alpha=5\%$) and the estimators are computed in
each subsample. Then, the minimum (maximum) of the six estimates is
returned as the lower (upper) end of the confidence interval.

The third issue is that many of our
estimators depend on
nuisance functions such as the outcome model $\mu(a, x)$ and the
conditional density $\pi(a | x)$. To avoid imposing restrictions on
the complexity of the nuisance function classes, we analyze estimators
based on cross-fitting. That is, unless otherwise stated, the nuisance
functions are assumed to be estimated from a different
sample than the sample used to compute the estimator. Such construction
can always be achieved by splitting the sample into $k$ folds; using
all but one fold for training the nuisance functions and the remaining
fold to compute the estimator. Then, the roles of the folds can be
swapped, thus yielding $k$ estimates that are
averaged to obtain a single estimate of the
parameter. For simplicity, we will use $k=2$, but our analysis
can be easily extended to the case where multiple splits are performed.

\section{Marginal Structural Models}
\label{section::MSM}

In this section we review
basic terminology and notation
for marginal structural models. 
We focus for now on studies with one time point;
we deal with time varying cases in Section~\ref{section::time}.
More detailed reviews can be found in  \cite{robins2009estimation} and \cite{hernan2010causal}.
Let
\begin{equation}\label{eq::msm1}
\E\{Y(a)\} \equiv \psi(a)  = g(a;\beta), \,\,\, \beta\in\mathbb{R}^k,
\end{equation}
be a model for the expected outcome under treatment regime $A = a$. An example is the linear model
$g(a;\beta)=b^T(a) \beta$
for some specified vector of basis functions
$b(a) = [b_1(a),\ldots ,b_k(a)]$.
It can be shown that $\beta$ in~\eqref{eq::msm1} satisfies the $k$-dimensional system of equations
\begin{equation}\label{eq::m}
\E\left[ h(A) w(A,X) \{ Y-g(A;\beta) \} \right]  = 0
\end{equation}
for any vector of functions $h(a) = [h_1(a),\ldots, h_k(a)]$,
where
$w(a,x)$ can be taken to be either
$1/\pi(a|x)$ or
$\pi(a)/\pi(a|x)$, and $\pi(a)$ is the marginal density of the treatment $A$. 
The latter weights are called stabilized weights and 
can lead to less variable estimators of $\beta$. We will use them throughout. The parameter $\beta$ can be estimated 
by solving the empirical analog of (\ref{eq::m}), leading to the estimating equations 
\begin{equation}\label{eq::ee}
\Pn \left[h(A)\widehat{w}(A, X) \{Y-g(A;\beta)\} \right] = 0,
\end{equation}
where
$\widehat{w}(a,x) = \hat \pi(a)/\hat\pi(a|x)$,
and $\hat \pi(a|x)$ and $\hat \pi(a)$ are 
estimates of $\pi(a|x)$ and $\pi(a)$.
Under regularity conditions, including the correct specification of $\pi(a|x)$,
confidence intervals based on
$\sqrt{n}(\hat\beta - \beta)\rightsquigarrow N(0,\sigma^2)$,
where
$\sigma^2 = M^{-1} \var[h(A)w(A,X)\{Y - g(A; \beta)\}] M^{-1}$ and $M = \E\{h(A)\nabla_\beta g(A; \beta)^T\}$, will be conservative.

Under model \eqref{eq::msm1}, 
every choice of $h(a)$ leads to a $\sqrt{n}$-consistent,
asymptotically Normal estimator of $\beta$, 
though different choices lead to different standard errors. 
If the MSM is linear, i.e. $g(a;\beta) = b(a)^T
\beta$, a common choice of $h(a)$ is $h(a) = b(a)$. 
In this case, the solution to the estimating equation (\ref{eq::ee}) can be obtained
by weighted regression,
$\widehat\beta = (B^T \mathbb{W} B)^{-1} B^T \mathbb{W} Y$,
where
$B$ is the $n\times k$ matrix with elements
$B_{ij} = b_j(A_i)$,
$\mathbb{W}$ is diagonal with elements
$\hat W_i \equiv \widehat{w}(A_i, X_i)$ and
$Y = (Y_1,\ldots, Y_n)$.

\section{Sensitivity Models}
\label{section::sens}

We now describe three models
for representing unmeasured confounding when treatments are continuous.
Each model defines a class of distributions
for $(U,X,A,Y)$
where $U$ represents unobserved confounders.
Our goal is to find bounds on causal quantities, such as $\beta$ or $g(a; \beta)$,
as the distribution varies over these classes.

\subsection{Propensity Sensitivity Model}

In the case of binary treatments $A\in\{0,1\}$,
a commonly used sensitivity model 
\citep{Rosenbaum} is
the odds ratio model
$$
\Gali \hspace{-0.08in} (\gamma) = 
\Biggl\{ \pi(a|x,u):\ 
\frac{1}{\gamma} \leq 
\frac{\pi(1|x,u)}{\pi(0|x,u)}
\frac{\pi(0|x,\tilde u)}{\pi(1|x,\tilde u)} \leq \gamma\ \ 
{\rm for\ all\ }u,\tilde{u},x \Biggr\}
$$
for $\gamma \geq 1$.
When $A$ is continuous,
it is arguably more natural to work with density ratios, and so we define
\begin{equation}\label{eq::Pi}
\Pi(\gamma) = 
\Biggl\{ \pi(a|x,u):\ \frac{1}{\gamma}  \leq \frac{\pi(a|x,u)}{\pi(a|x)} \leq \gamma ,\ 
\int \pi(a|x,u) da = 1,\ {\rm for\ all\ }a,x,u \Biggr\}.
\end{equation}
We can think of $\Pi(\gamma)$ as defining a neighborhood around
$\pi(a|x)$.  This is related to the class in
\cite{tan2006distributional} but we consider density ratios rather
than odds ratios. There are other constraints possible, such as $\int
\pi(a | x, u) d\P(u | x) = \pi(a | x)$; we leave enforcing
these additional constraints, which can yield more precise bounds, for
future work.

\subsection{Outcome Sensitivity Model}
\label{section::os}

For an outcome-based sensitivity model, we define a 
neighborhood around $\mu(x,a)$ given by
$$
{\cal M}(\delta) = \left\{\mu(u,x,a):\  |\Delta(a)| \leq \delta, \,\,\, \Delta(a) = \int [\mu(u,x,a)-\mu(x,a)] d\P(x,u)  \right\}, 
$$ 
which is the set of unobserved outcome regressions (on measured
covariates, treatment, and unmeasured confounders) such that
differences between unobserved and observed regressions differ by at
most $\delta$ after averaging over measured and unmeasured covariates.
We immediately have the simple nonparametric bound 
$\E\{\mu(a, X)\}-\delta \leq \E\{Y(a)\} \leq \E\{\mu(a, X)\}+\delta$.  
For a given $\Delta(a)$,
is a known function, nonparametric bounds can be computed by
regressing an estimate of 
$\Delta(A) + w(A,X)\{Y - \mu(A, X)\} + \int \mu(A, x) d\Pb(x)$ 
on $A$ (see, e.g. \cite{kennedy2017non},
\cite{semenova2021debiased}, \cite{foster2019orthogonal},
\cite{bonvini2020fast}).  However, our main goal is not to bound $\E \{ Y(a) \}$, but bound the
parameters $\beta$ of the MSM or the MSM itself.  Finding these bounds
under outcome sensitivity will require specifying an outcome
model. For the propensity sensitivity model, we will also need an
outcome model if we want doubly robust estimators of $\beta$.

\subsection{Subset Confounding}
\label{section::subset}

\cite{bonvini2021sensitivity} consider a model where only an unknown fraction
$\epsilon$ of the population is subject to unobserved
confounding. Specifically, 
suppose there exists a latent binary variable $S$ such that
$P(S=0)=\epsilon$ as well as
$Y(a) \ind A  |  X, S = 1$
and $Y(a) \ind A  |  X, U, S = 0$.
It follows that
$P = (1-\epsilon) P_1 + \epsilon P_0$
where $P_j$ is the distribution of $(U,X,A,Y)$ given $S=j$.
For the $S = 0$ group of units, we will control the extent of
unmeasured confounding using either the outcome model or the
propensity sensitivity model. 
This can be regarded as a type of contamination model.

Results under the propensity and outcome 
sensitivity confounding models are in the next two sections.
Due to space restrictions, the results on 
subset confounding are in the appendix.

\section{Bounds under the Propensity Sensitivity Model \label{section::MSMPS}}
\subsection{Preliminaries \label{sec::prelim}}

In this section, we develop preliminary results needed
to derive bounds under the propensity sensitivity model.
A preliminary step in deriving bounds for the MSM
is to first bound $\E\{Y(a)|X\}$
and it may be verified that
$\E \{ Y(a)|X \} = m(a,X)$
where
$m(A, X) = \E\{Y v(Z)  |  A, X\}$ 
and
\begin{align*}
v(Z) \equiv \E\left\{ \frac{\pi(A  |  X)}{\pi(A  |  X, U)} \Biggm| A, X, Y\right\} \in \left[\gamma^{-1}, \gamma\right].
\end{align*}
It is easy to see that $\E\{v(Z)  |  A, X\} =1$.
So bounding $\E \{ Y(a)|X \}$
is equivalent to bounding
$m(a,X)= \E\{Y v(Z)  |  A, X\}$ 
as $v$ varies over the set
\begin{equation}
\label{eq::firstV}
{\cal V}(\gamma) = 
\Biggl\{ v(\cdot):\ \gamma^{-1} \leq v(z) \leq \gamma,\ \E\{v(Z)  |  X=x,A=a\} =1\ {\rm for\ all\ }x,a\Biggr\}.
\end{equation}

\begin{proposition}
\label{prop::msm}
The following moment condition holds for the MSM:
\begin{align} \label{eq::general_mc_noU}
\E\left\{ h(A) \Bigl[\int m(A, x) d\Pb(x) - g(A; \beta)\Bigr] \right\} = 
\mathbb{U}\left[h(A_1)\{m(A_1, X_2) - g(A_1; \beta)\} \right]= 0,
\end{align}
where $(X_1,A_1)$ and $(X_2,A_2)$ 
are two independent draws (see Section~\ref{sec::issues}). 
\end{proposition}

Notice that if $U = \emptyset$, then $v(Z) = 1$ and $m(a, x) = \mu(a,
x) = \E\{Y  |  A = a, X = x\}$. However, when there is residual
unmeasured confounding, $m(a, x)$ does not equal $\E(Y  |  A = a, X =
x)$ and in general cannot be identified. However, it can still be
bounded under the propensity sensitivity model, as in the following
lemma.

\begin{lemma}\label{lemma:np_cond_bounds}
For $j \in \{\ell, u\}$
(corresponding to lower and upper bound)
let $q_j(Y |  A, X)$ 
denote the $\tau_j$-quantile of $Y$ 
given $(A, X)$, where $\tau_\ell = 1 / (1 + \gamma)$ and $\tau_u = \gamma / (1 + \gamma)$. Define 
\begin{align*}
v_\ell(Z) = \gamma^{\sgn\{q_\ell(Y |  A, X) - Y\}} \quad \text{ and } v_u(Z) = \gamma^{\sgn\{Y - q_u(Y  |  A, X)\}}.
\end{align*}
Then 
$m_\ell(a,x) \leq m(a,x) \leq m_u(a,x)$, where
$m_j(a, x) = \E\left\{Yv_j(Z)  |  A = a, X = x \right\}$, $ j \in \{ u, \ell\}$.
\end{lemma}

Now that we have bounds on $m(a,x)$, we turn to finding bounds on the MSM 
$g(a; \beta)$ and on its parameter $\beta$.
We will use the notation $c_\ell = \gamma^{-1}$, $c_u = \gamma$,
$S_j \equiv s(Z; q_j) = q_j(Y  |  A, X) + \{Y - q_j(Y  |  A, X)\} c_j^{\sgn\{Y - q_j(Y  |  A, X)\}}$,
$\kappa_j \equiv \kappa(A, X; q_j) = \E\{ S_j  |  A, X\}$ and 
\begin{equation}\label{eq::vp1}
\varphi_j(Z_1, Z_2) \equiv \varphi_j(Z_1, Z_2; w, q_j, \kappa_j) = 
w(A_1, X_1)\{s(Z_1; q_j) - \kappa(A_1, X_1; q_j)\} + \kappa(A_1, X_2; q_j).
\end{equation}
Notice that
\begin{equation}\label{eq::ntt}
\mathbb{U}\{\kappa(A_1, X_2; q_j)\} = \int \int m_j(a, x) d\Pb(a) d\Pb(x),
\end{equation}
since $\E\left[c_j^{\sgn\{Y - q_j(Y  |  A, X)\}}  |  A, X\right] = 1$.

\subsection{Bounds on $g(a;\beta)$ \label{sec:bounds_g}}

Under the MSM $\E \{ Y(a) \}= g(a;\beta)$, given the discussion in
Section~\ref{sec::prelim}, we have that 
$\E \{ Y(a) \} = \E\{m(a,X)\}$ if $Y(a) \ind A | (X, U)$.  This implies that
$\E\{m_\ell(a, X)\} \leq g(a; \beta)\leq \E\{m_u(a, X)\}$,
where $m_\ell$ and $m_u$ are defined in Lemma \ref{lemma:np_cond_bounds}. Thus, a
straightforward way to bound $g(a; \beta)$ is to assume that
the bounds follow a model similar to the model we assume under
no unmeasured confounding, when $\E \{ Y(a) \} = g(a; \beta)$ is
identified. 
That is, we let $\E\{m_j(a, X)\} = g(a; \beta_j)$, $j  \in \{ u, \ell\}$, 
and estimate $\beta_j$ by solving the empirical analog of the moment condition:
\begin{align} 
\label{eq::general_mc}
\E\left\{ h(A) \Bigl[\int m_j(A, x) d\Pb(x) - g(A; \beta_j)\Bigr] \right\} = 0, \quad j \in \{ u, \ell\}.
\end{align}
Using~(\ref{eq::ntt}) and the fact that the first term in~(\ref{eq::vp1}) has conditional mean 0,
we also have that
$\mathbb{U}\left[h(A_1)\left\{\varphi_j(Z_1, Z_2) - g(A_1; \beta_j)\right\} \right]= 0.$
Given an estimate of the function $\varphi_j(Z_1, Z_2)$ in~(\ref{eq::vp1}),
estimated from an independent sample, we 
estimate $\beta_j$ by solving
\begin{align}
\label{eq::general_mc2}
\mathbb{U}_n\left[h(A_1)\left\{\widehat\varphi_j(Z_1, Z_2) - g(A_1; \widehat\beta_j)\right\} \right]= 0.
\end{align}
The following proposition provides the asymptotic distributions of
$g(a; \betahat_j)$, $j \in \{ u, \ell\}$.

\begin{proposition}\label{prop:vdv}
Suppose the following conditions hold:
\begin{enumerate}
\item The function class $\mathcal{G}_l = \left\{ a \mapsto h_l(a)g(a; \beta)\right\}$ is Donsker 
for every $l = \{1, \ldots, k\}$ with integrable envelop
and $g(a;\beta)$ is a continuous function of $\beta$.
\item 
For $j\in\{\ell,u\}$,
the map $\beta \mapsto \mathbb{U}\{h(A)[\varphi_j(Z_1, Z_2) - g(A_1; \beta)]$
is differentiable at all $\beta$ with continuously 
invertible matrices $\dot\Psi_{\beta_0}$ and $\dot\Psi_{\widehat\beta}$, where $\dot\Psi_{\beta} = -\E\{h(A)\nabla^Tg(a; \beta)\}$;
\item $\left\|\int S_2\left\{ \widehat\varphi_j(Z_1, z_2) - \varphi_j(Z_1, z_2) \right\} d\Pb(z_2) \right\| = o_\Pb(1)$;
\item $\Norm{w - \widehat{w}} \Norm{\kappa_j- \widehat\kappa_j} + \Norm{q_j - \qhat_j}^2 = o_\Pb(n^{-1/2})$, 
where $\varphi_j$, $\kappa_j$ and $q_j$ are defined in Section~\ref{sec::prelim}.
\end{enumerate}
Then
$\sqrt{n}(\betahat_j - \beta_j) 
\indist N\left(0, 4\var\{\dot\Psi_{\beta_j}^{-1} \phi_j(T; \beta_j)\} \right), $ 
where $\phi_j(Z_1; \beta_j) = \int S_2 h(A_1)\left\{\varphi_j(Z_1, z_2) - g(A_1; \beta_j)\right\} d\Pb(z_2)$,
and it follows that
$$
\sqrt{n} \{ g(a;\hat\beta_j)-g(a;\beta_j) \}\rightsquigarrow 
N\left(0, 4\nabla g(a; \beta_j)^T \var\{\dot\Psi_{\beta_j}^{-1} \phi_j(Z; \beta_j)\} \nabla g(a; \beta_j)\right), \quad j \in \{ u, \ell\}.
$$
\end{proposition}

The main requirement, in condition (d), to achieve asymptotic normality is that certain
products of errors for estimating the nuisance functions are
$o_\Pb(n^{-1/2})$. This can be achieved even if these functions are
estimated at nonparametric rates, e.g. $n^{-1/4}$, under structural
constraints such as smoothness or sparsity. We note that, strictly
speaking, our estimator is not doubly robust since one needs to
consistently estimate $q_j$ for consistency. However, the dependence
on the estimation error in $\widehat{q}_j$ is still second-order, in
that it depends on the squared error.

\subsection{Bounds on $g(a;\beta)$ when $g(a;\beta)$ is linear \label{sec:bounds_glinear}}

When the MSM is linear, it is straightforward to bound $g(a; \beta)=b(a)^T\beta$ directly, 
without assuming that the bounds themselves
follow parametric models $g(a; \beta_j)$. Let $h(A) = b(A)$ and 
$Q =\E\{b(A)b(A)^T\}$. Then we can re-write
$g(a; \beta) = b(a)^TQ^{-1}\mathbb{U}\left\{b(A_1)m(A_1, X_2)\right\}.$ 
Let  $\lambda^{-}(a, A) = \one\{b(a)^TQ^{-1} b(A) \leq 0\}$ and $\lambda^{+}(a, A) = 
\one\{b(a)^TQ^{-1} b(A) \geq 0\}$. Further define
\begin{align*}
g^s_j(a) = \mathbb{U} \{ b(A_1)\lambda^s(a, A_1)\kappa(A_1, X_2; q_j)\}
\end{align*}
for $s = \{-, +\}$ and $j = \{\ell, u\}$. 
Bounds on $g(a; \beta) = b(a)^T \beta$ are
$g_\ell(a) \leq g(a;\beta) \leq g_u(a)$ where
$g_\ell(a) = b(a)^TQ^{-1} \{g^{+}_\ell(a) + g^{-}_u(a)\}$
and
$g_u(a) = b(a)^TQ^{-1} \{g^{-}_\ell(a) + g^{+}_u(a)\}$.
That is, depending on the sign of $b(a)^T Q^{-1}b(A)$, we set $m(A,
x) = m_\ell(A, x)$ or $m(A,x) = m_u(A, x)$. 
Let
$$
f^s_j(Z_1, Z_2) = \lambda^{s}(A_1)\varphi_j(Z_1, Z_2).
$$
We analyze the performance of
estimators that construct $\widehat{f}^s_j(Z_1, Z_2)$ from a separate,
independent sample and output $\widehat{g}_j(a_0) = b(a_0)^T
\widehat\beta_j$, where
\begin{align*}
\widehat\beta_\ell &= \argmin_{\beta \in \R^k} \mathbb{U}_n \left\{\widehat{f}^{-}_u(Z_1, Z_2) + \widehat{f}^{+}_\ell(Z_1, Z_2) - b(A_1)^T\beta\right\}^2\\
\widehat\beta_u &= \argmin_{\beta \in \R^k} \mathbb{U}_n \left\{\widehat{f}^{-}_\ell(Z_1, Z_2) + \widehat{f}^{+}_u(Z_1, Z_2) - b(A_1)^T\beta\right\}^2.
\end{align*}
The following proposition gives the limiting distribution of the estimated upper and lower bounds 
$\widehat{g}_j(a)$ for $g(a; \beta)$, $j \in \{u, \ell\}$.
\begin{proposition} \label{prop::glinear}
Suppose the following conditions hold:
\begin{enumerate}
	\item $\left\| \int S_2\{\widehat{f}_j^s(Z_1, z_2) - f_j^s(Z_1, z_2)\} d\Pb(z_2)\right\| = o_\Pb(1)$;
	\item $\|\widehat{q}_j - q_j\|^2 + \|\widehat{w} - w\| \|\widehat\kappa_j- \kappa_j\| = o_\Pb(n^{-1/2})$;
	\item The density of $b(a)^T Q^{-1} b(A)$ is bounded.
\end{enumerate}
Then 
$$
\sqrt{n}\{\widehat{g}_j(a) - g_j(a)\} \rightsquigarrow 
N \left(0, 4b(a)^T\var\left[ Q^{-1} \int S_2b(A_1)\{f_j^s(Z_1, z_2) - b^T(A_1)\beta_j\} d\Pb(z_2) \right] b(a) \right).
$$
\end{proposition}

Another approach for getting bounds on $g(a; \beta)$ is to note that
$\delta g(a;\beta)/\delta v = \sum_j b_j(a) \delta \beta_j /\delta v$
(where $\delta$ is the functional derivative)
and then apply the homotopy algorithm from Section~\ref{section::homotopy}.

\subsection{Bounds on $\beta$ \label{section::homotopy}}

We now turn to finding approximate bounds on components of
$\beta$ rather than on $g(a;\beta)$. Suppose, to be concrete, that we want to upper bound $\beta_1$.  (Lower bounds can be
found similarly.)  At this point, we re-name ${\cal V}(\gamma)$
in~(\ref{eq::firstV}) as
${\cal V}_{\rm small}(\gamma)$ and we define
\begin{align*}
{\cal V}_{\rm large}(\gamma) &= \Biggl\{ v(\cdot):\ \gamma^{-1} \leq v(z) \leq \gamma,\ \E[v(Z)] = 1\Biggr\}.
\end{align*}
Bounds over ${\cal V}_{\rm large}(\gamma)$ are conservative
but, as we shall see, are easier to compute.
Next we define two functionals.
Let $F_1(v)$ be the value of $b$ that solves
$$
\int y h(a) w(a,x) v(z) \Pb(z) = \int h(a) w(a,x)  g(a;b) v(z) \Pb(z)
$$
and
$F_2(v)$ be the value of $b$ that solves
$$
\hspace{-.21in} \int y h(a) w(a,x) v(z) \Pb(z) = \int h(a) w(a,x)  g(a;b) \Pb(z).
$$
At the true value $v_*$ we have
$\beta_{1*} =  e^T F_1(v_*) = e^T F_2(v_*)$
where 
$e=(1,0,\ldots, 0)$ and
$\beta_{1*}$ is the true value of $\beta_1$.
But $F_1(v)\neq F_2(v)$ in general, and
bounding $F_1(v)$ and $F_2(v)$ both lead to valid bounds for $\beta_1$.
A quick summary of what will follow is this:
\begin{enumerate}
\item[i.] For ${\cal V}_{\rm small}(\gamma)$, bounds based on $F_1$ and $F_2$ are equal, as stated in Lemma~\ref{lemma:F1F2}.
These bounds require quantile regression.
\item[ii.] For ${\cal V}_{\rm large}(\gamma)$, bounds based on $F_1$ and $F_2$ are different so we 
take their intersection.
These bounds do not require quantile regression. In our experience, bounds based on $F_1$ are often tighter.
\end{enumerate}

\begin{lemma}\label{lemma:F1F2}
We have
$\inf_{v\in {\cal V}_{\rm small}(\gamma)}e^T F_1(v) = \inf_{v\in {\cal V}_{\rm small}(\gamma)}e^T F_2(v)$
and
$\sup_{v\in {\cal V}_{\rm small}(\gamma)}e^T F_1(v) = \sup_{v\in {\cal V}_{\rm small}(\gamma)}e^T F_2(v)$.
For ${\cal V}_{\rm large}(\gamma)$, the bounds may differ.
\end{lemma}

We want to find $v_\gamma$ such that
$e^T F_k(v_\gamma) = \sup_{v\in {\cal V}}e^T F_k(v)$, for $k \in \{1,2\}$ and
${\cal V} \in \{ {\cal V}_{\rm small}, {\cal V}_{\rm large} \}$.

Unless the MSM $g(a; \beta) = b^T(a)\beta$ is linear, determining the optimal $v_\gamma$ is intractable, so we find an approximate bound. For example, to optimize over $\cal V_{\rm large}(\gamma)$, we proceed as follows:
\begin{enumerate}
\item We will find a function $v_\gamma$ that is a local maximum of $F_k(v)$.
\item We show that 
$v_\gamma$ is defined by a fixed point equation 
$v_\gamma = L(v_\gamma)$.
\item We construct an increasing grid
$\{\gamma_1,\gamma_2, \ldots,\}$
where $\gamma_1 = 1$ and
$\gamma_j = \gamma_{j-1}+\delta$.
Then we take
$v_{\gamma_j} \approx L(v_{\gamma_{j-1}})$.
\item
In the limit, as $\delta\to 0$,
this defines a sequence of functions $(v_\gamma:\ \gamma \geq 1)$
where each $v_\gamma$ is a local optimizer in ${\cal V}_{\rm large}(\gamma)$.
\end{enumerate}

We refer to this as a {\em homotopy algorithm}.
(An alternative approach based on gradient ascent
is described Appendix~\ref{app:coord}.)
To make this precise, we need the functional derivative of $F_k(v)$ with respect to $v$.
First we recall the definition of a functional derivative:
if $G(v)\in\mathbb{R}$,
we say that
$\frac{\delta G(v)}{\delta v}$ is the functional derivative
of $G(v)$ with respect to $v$ in $L_2(\Pb)$ if 
$$
\int \frac{\delta G(v)}{\delta v} (z) f(z) d\Pb(z) dz = 
\left[ \frac{d}{d\epsilon} G[ v + \epsilon f]\right]_{\epsilon=0}
$$
for every function $f$.
When $G(v) = (G_1(v),\ldots, G_k(v))$
is vector valued,
we define
$\delta G/\delta v = (\delta G_1(v)/\delta v,\ldots, \delta G_k(v)/\delta v)$.

\begin{lemma}[Functional derivatives]
\label{lemma:beta_deriv}
We have
\begin{align}\label{eq::fd}
\frac{\delta F_1(v)}{\delta v} (z) &=
\Biggl\{ 
\E\Bigl[v(Z)h(A)w(A,X)\nabla_\beta g(A; \beta)^T \Bigr]\Biggr\}^{-1} h(a) (y-g(a;\beta)) w(a,x),\\
\nonumber
\frac{\delta F_2(v)}{\delta v} (z)&=
\Biggl\{\E\Bigl[ h(A)w(A,X)\nabla_\beta g(A; \beta)^T \Bigr]\Biggr\}^{-1} h(a) y w(a,x).
\end{align}
\end{lemma}

Notice that, unless $g(a; \beta)$ is linear in $\beta$, $\frac{\delta F_2(v)}{\delta v} (z)$ depends on $v(z)$ through $\nabla_\beta g(A; \beta)$ since $\beta$ is implicitly a function of $v(z)$. We can now find the expression for the local optimizer $v_\gamma$ from Step (a) above.

\begin{lemma}
\label{lemma::FP}
Suppose that for every $v$,
$(\delta F_k(v)/\delta v) (Z)$
has a continuous distribution.
There is a set of functions $(v_\gamma:\ \gamma\geq 1)$
such that:

1. $v_\gamma\in {\cal V}_{\rm large}(\gamma)$;

2. $v_\gamma$ satisfies the fixed point equation
\begin{equation}\label{eq::vgamma}
v_{\gamma}(z) = 
\gamma \one \Bigl[ d_{\gamma}(z) \geq q_u(d_\gamma) \Bigr] +
\gamma^{-1} \one \Bigl[d_{\gamma}(z) <  q_u(d_\gamma) \Bigr]
\end{equation}
where
$$
d_{\gamma} = e^T\left(\frac{\delta F_k(v)}{\delta v}\Biggr|_{v = v_{\gamma}}\right)
$$
and $q_u(d_\gamma)$ is the $\tau_u = \gamma / (1 + \gamma)$ quantile of
$d_\gamma(Z)$.
(This is a fixed point equation since $d_\gamma$ on the right hand 
side is a function of $v_\gamma$.);

3. $v_\gamma$ is a local maximizer of $e^TF_k(v)$, in the sense that,
for all small $\epsilon>0$,
$e^T F_k(v_\gamma) \geq e^T F_k(v) + O(\epsilon^2)$
for any $v\in {\cal V}_{\rm large}(\gamma)\bigcap B(v_\gamma,\epsilon)$
where, for any $v$ and any $\epsilon>0$ we define
$B(v,\epsilon) = \{ f:\ \int (f-v)^2 d\Pb(z) \leq  \epsilon^2\}.$ 
\end{lemma}


In practice, we compute $v_\gamma$ sequentially
using an increasing sequence of values of $\gamma$.
Using~(\ref{eq::vgamma}) we approximate $v_\gamma$ by
$\gamma \one \{ d_{\gamma-\delta}(z) \geq  q_u(d_{\gamma-\delta}) \} +
\gamma^{-1}  \one \{ d_{\gamma-\delta}(z) <   q_u(d_{\gamma-\delta})  \}$
where $ q_u(d_{\gamma-\delta})$ is the
$\tau_u = \gamma/(1+\gamma)$ quantile of
$d_{\gamma-\delta}(Z)$
and $\delta$ is a small positive number.
The sample approximation to the functional derivative for observation $i$ is
\begin{equation}\label{eq::di1}
d_i = \frac{\partial \hat F_1(v)}{\partial V_i} =
\left\{\frac{1}{n}\sum_j h(A_j)V_j\hat W_j \nabla_\beta g(A_j;\hat\beta)^T \right\}^{-1}
h(A_i)\hat W_i (Y_i- g(A_i;\hat\beta))
\end{equation}
for $F_1$ and
\begin{equation}\label{eq::di}
d_i = \frac{\partial \hat F_2(v)}{\partial V_i} =
\left\{\frac{1}{n}\sum_j h(A_j)\hat W_j \nabla_\beta g(A_j;\hat\beta)^T\right\}^{-1}
h(A_i)\hat W_i Y_i
\end{equation}
for $F_2$,
where $V_i = v(X_i,A_i,Y_i)$. 
The algorithm is described in Appendix~\ref{app::homo}.
The lower bound on $\beta_1$ is obtained the same way, with (\ref{eq::vgamma}) replaced by
$v_{\gamma}(z) = \gamma^{-1} \one \{ d_{\gamma}(z) \geq   q_\ell (d_{\gamma})  \} +
\gamma  \one \{ d_{\gamma}(z) <   q_\ell (d_{\gamma}) \}$,
where 
$\tau_\ell = 1 / (1 + \gamma)$.  
Getting confidence intervals for these bounds is challenging
because 
we need to adjust the estimator with the influence function to make the
bias second order, but
their influence functions are very complicated;
the details are in Appendix~\ref{app::influ}.


For ${\cal V} = {\cal V}_{\rm small}$, which imposes the stronger restriction
$\E\{v(Z)  |  A,X\} = 1$, we replace $q_u (d_{\gamma})$ in~(\ref{eq::vgamma}) with $q_u (d_\gamma  |  A,X)$, the conditional quantile of $d_\gamma(z)$ given $(X,A)$. Then
\begin{align*}
d_\gamma(Z) = \E\{h(A)\nabla_\beta g(A; \beta)^T\}^{-1}  h(A) w(A,X) Y \equiv T(A, X) Y,
\end{align*}
so that the $\tau$th quantile of $d_\gamma(Z)$ given $(A, X)$ can be expressed as
\begin{align*}
q_\tau(d_\gamma  |  A,X) = \begin{cases}
T(A, X) q_\tau (Y  |  A,X) & \text{ if } T(A, X) < 0,\\
T(A, X) q_{1-\tau} (Y  |  A,X) & \text{ if } T(A, X) > 0,
\end{cases}
\end{align*}
where $q_\tau (Y  |  A,X)$ is the $\tau$th quantile of $Y$ given
$(A, X)$. Then, to obtain an upper bound on $\beta_1$,
$v_\gamma(Z)$ has to satisfy the fixed-point equation:
\begin{align*}
v_\gamma(z) = \one\{e^T T(a, x) \geq 0\} v_u(z) + \one\{e^T T(a, x) < 0\} v_\ell(z),
\end{align*}
where $v_u(Z) = \gamma^{\sgn\{Y - q_u(Y  |  A,X)\}}$ and $v_\ell(Z)
= \gamma^{\sgn\{q_\ell(Y  |  A,X) - Y\}}$ are defined in Lemma
\ref{lemma:np_cond_bounds}, and $T(a, x)$ depends on $v_\gamma(z)$ through $\beta$. Similarly, a lower bound on $\beta_1$
requires $v_\gamma(z)$ to satisfy
\begin{align*}
v_\gamma(z) = \one\{e^TT(a, x) \le 0\} v_u(z) + \one\{e^TT(a, x) > 0\} v_\ell (z).
\end{align*}

\subsection{Bounds on $\beta$ when $g(a; \beta)$ is linear \label{sec::betalin}}

If $g(a;\beta)=b(a)^T \beta$, we can derive simpler bounds. In this case we have
$$
F_1(v) =  \int y w(a,x) v(z) M^{-1}(v) b(a) d\Pb(z),\ \ \ F_2(v) =  \int y w(a,x) v(z) M^{-1} b(a) d\Pb(z)
$$
where
$M(v) = \int w(a,x) v(z) b(a)b(a)^T d\Pb(z)$
and $M = \int w(a,x) b(a)b(a)^T d\Pb(z)$.

\begin{lemma} \label{lemma:linB}
Let $f(z) = y w(a,x)  e^T M^{-1} b(a)$.
We have
\begin{align*}
\inf_{v\in {\cal V}_{\rm small}(\gamma)}e^T F_1(v) &= 
\inf_{v\in {\cal V}_{\rm small}(\gamma)}e^T F_2(v) = \int f(z) \underline{v}(z) dP(z),\\
\sup_{v\in {\cal V}_{\rm small}(\gamma)}e^T F_1(v) &= 
\sup_{v\in {\cal V}_{\rm small}(\gamma)}e^T F_2(v) = \int f(z) \overline{v}(z) dP(z),
\end{align*}
where    \vspace{-0.1in}
\begin{align*}
\overline{v}(Z) &= \gamma \one \{ f(Z) \geq q_u(f  |  A,X) \} + \gamma^{-1} \one \{ f(Z) < q_u(f  |  A,X) \},\\
\underline{v}(z)&= \gamma \one \{ f(Z) \leq q_\ell(f  |  A,X) \} + \gamma^{-1} \one \{ f(Z) > q_\ell (f  |  A,X) \},
\end{align*}
and
$q_u(f  |  A,X)$ and $q_\ell( f  |  A,X)$ are the $\tau_u = \gamma/(1+\gamma)$ and $\tau_\ell = 1/(1+\gamma)$ quantiles of $f(Z)$ given $(X,A)$.
\end{lemma}

Again, for the class
${\cal V}_{\rm large}(\gamma)$
the bounds can differ
and one can construct examples where either of the two is tighter, 
so we use the intersection of the bounds from $F_1$ and $F_2$.
Bounding $F_2(v)$ 
over ${\cal V}_{\rm large}(\gamma)$ is straightforward
as discussed in the following lemma.

\begin{lemma}\label{lemma::largeiseasy}
Let $f(z) = y w(a,x)  e^T M^{-1} b(a)$.
Then
\begin{align*}
\inf_{v\in {\cal V}_{\rm large}(\gamma)} F_2(v) &= F_2(\underline{v}),\ \ \ \ 
\sup_{v\in {\cal V}_{\rm large}(\gamma)} F_2(v) = F_2(\overline{v}),
\end{align*}
where  \vspace{-0.1in}
\begin{align*}
\overline{v}(Z) &= \gamma \one \{f(Z) \geq q_u(f)\} + \gamma^{-1} \one (f(Z) < q_u(f)), \\
\underline{v}(z) &= \gamma \one \{ f(Z) \leq q_\ell(f) \} + \gamma^{-1} \one \{ f(Z) > q_\ell (f) \},
\end{align*}
and
$q_u(f)$ and $q_\ell(f)$ are the $\tau_u = \gamma/(1+\gamma)$ and $\tau_\ell = 1/(1+\gamma)$ quantiles of $f(Z)$.
\end{lemma}

That is, we only need marginal quantiles for the bound on $F_2(v)$.
We do not have a closed form expression
for bounds on $F_1(v)$ over ${\cal V}_{\rm large}(\gamma)$.
Instead we use the homotopy algorithm.
As in the general MSM case presented in Section~\ref{section::homotopy}, getting
confidence intervals for the bounds of $\beta_1$ over ${\cal V}_{\rm  large}$ is challenging because their influence functions involve
solving an integral equation.


\subsection{Local (Small $\gamma$) Bounds on $\beta$ \label{sec::local}}

A fast, simple approach to bounding $F_1$ over ${\cal V}_{\rm large}(\gamma)$ 
is based on a functional expansion of $F_1(v)$ around the function $v_0=1$, or
alternatively, an expansion of $F_1( L )$ around the function $L_0 \equiv \log v_0=0$.
In principle, this will lead to tight bounds only
for $\gamma$ near 1, but, in our examples,
it leads to accurate bounds over a range of $\gamma$ values;
see Figures~\ref{fig::beta},
\ref{fig::simulated} and 
\ref{fig::explain}.
Note that we do not need local bounds based on $F_2$
because we have an exact expression in that case.

Let $L (z) = \log v(x,a,y)$.
Our propensity sensitivity model is the set of functions $L$ such that
$|| L ||_\infty\leq \log \gamma$.
Note that
$\frac{\delta F_1}{\delta L} (z)  = \frac{\delta F_1}{\delta v} (z) e^{L} = \frac{\delta F_1}{\delta v} (z) v(z)$.
No unmeasured confounding corresponds to 
$v_0(z) = 1$,  $L_0(z) = 0$ and $\gamma=1$.
Then
$F_1(L) = F_1(L_0) + e^T \int (L-L_0) \frac{\delta F_1}{\delta L} (z) d\Pb(z) + O(\gamma-1)^2 = 
F_1(L_0) + e^T \int L  \frac{\delta F_1}{\delta v} (z)  d\Pb(z)  + O(\gamma-1)^2$
where $F_1(L_0)$ is the value of $\beta_1$ assuming no unmeasured confounding.
Now, by Holder's inequality,
$\int L \frac{\delta F_1}{\delta v} (z) d\Pb  \leq
|| L - L_0||_\infty \int | \frac{\delta F_1}{\delta v} (z)  d\Pb | \leq
\log\gamma \int | \frac{\delta F_1}{\delta v} (z)  d\Pb |.$
So, up to order $O(\gamma-1)^2$,
\begin{equation}
\beta_1(L_0) - \log\gamma \int \Biggl| \frac{\delta F_1}{\delta v} (z)  d\Pb \Biggr| \leq
F_1(L)\leq 
\beta_1(L_0) +  \log\gamma \int \Biggl| \frac{\delta F_1}{\delta v} (z)  d\Pb \Biggr|.
\end{equation}

\section{Bounds under the Outcome Sensitivity Model}

\label{section::MSMOS}

Consider now the
outcome sensitivity model
from Section~\ref{section::os}.
Recall that
$\mu(A, X, U) = \E(Y|A, X, U)$ is the outcome regression on treatment and both observed and unobserved confounders, and
$\Delta(a) = \int \{\mu(a, x, u) - \mu(a, x)\} d\Pb(x, u)$ is the (integrated) difference between this regression and its observed counterpart, 
and $|\Delta(a)|\leq \delta$.
If $Y(a) \ind A | (X, U)$, then
\begin{align*}
\E \{ Y(a) \} = \int \mu(a, x, u) d\Pb(x, u) = \Delta(a) + \int \mu(a, x) d\Pb(x).
\end{align*}

We can write a corresponding MSM moment condition as
\begin{align*}
\E\left[ h(A)  \left\{\Delta(A) + \int \mu(A, x) d\Pb(x) - g(A; \beta) \right\} \right]= 0
\end{align*}
so that $\beta$ is identified under no unmeasured confounding whenever
$\E\{h(A)\Delta(A)\} = 0$. Using an approach similar to Section \ref{sec:bounds_g}, if
we assume that the bounds $\int \mu(a, x) d\Pb(x) \pm \delta$ follow
models $g(a; \beta_\ell)$ and $g(a; \beta_u)$, it is straightforward to
estimate $\beta_\ell$ and $\beta_u$ by solving the empirical, influence function based, bias-corrected analogs of the moment
conditions
\begin{align*}
& \E\left[h(A)\left\{ \int \mu(A, x) d\Pb(x) + \delta - g(A; \beta_u) \right\}\right]  = 0,\\
& \E\left[h(A)\left\{ \int \mu(A, x) d\Pb(x) - \delta - g(A; \beta_\ell) \right\}\right]  = 0.
\end{align*}
Inference can be performed as outlined in Proposition~\ref{prop:vdv}. 

In the linear MSM case, we have $g(a; \beta) = b(a)^T \beta =
b(a)^TQ^{-1} \mathbb{U}\left[ h(A_1)\left\{ \mu(A_1, X_2) +  \Delta(A_1) \right\} \right]$, 
where $Q = \E\{h(A)b(A)^T\}$.  Therefore, valid bounds on $g(a; \beta)$ are
$
 b(a)^TQ^{-1} \mathbb{U}\left\{b(A_1) \mu(A_1, X_2) \right\} \pm \delta\E|b(a)^TQ^{-1}b(A)|,
 $
which we re-write as
\begin{align*}
	& g_\ell(a) =  b(a)^TQ^{-1} \mathbb{U}\left(b(A_1) \left[ \mu(A_1, X_2) - \delta \sgn\left\{b(a)^TQ^{-1}b(A_1)\right\} \right] \right) \\
	& g_u(a) =  b(a)^TQ^{-1} \mathbb{U}\left(b(A_1) \left[ \mu(A_1, X_2) + \delta \sgn\left\{b(a)^TQ^{-1}b(A_1)\right\} \right] \right),
\end{align*}
since
$
|b(a)^TQ^{-1}h(A)| = b(a)^TQ^{-1} \sgn\left\{b(a)^TQ^{-1}b(A)\right\} b(A).
$

Our estimators are
$$
\widehat{g}_j(a) = b(a)^T \widehat\beta_j,  \quad \widehat\beta_j = 
\argmin_{\beta \in \R^k} \mathbb{U}_n \left\{ \widehat\zeta_j(Z_1, Z_2) - b(A_1)^T\beta \right\} ^2, \quad j \in \{ u, \ell\},
$$
where  \vspace{-0.1in}
\begin{align*}
	& \zeta_\ell(Z_1, Z_2) = w(A_1, X_1)\{Y_1 - \mu(A_1, X_1)\} + \mu(A_1, X_2) - \delta \sgn\left\{b(a)^TQ^{-1}b(A_1)\right\} ,\\
	& \zeta_u(Z_1, Z_2) = w(A_1, X_1)\{Y_1 - \mu(A_1, X_1)\} + \mu(A_1, X_2) + \delta \sgn\left\{b(a)^TQ^{-1}b(A_1)\right\}. \\
\end{align*}

To simplify the analysis of our estimators and avoid imposing additional Donsker-type requirements on $\widehat\mu$ and $\widehat\pi$, we proceed by assuming that $\widehat\zeta_l$ and $\widehat\zeta_u$ are estimated on samples independent from that used to compute the U-statistic in the empirical risk minimization step. This means that, in finite samples, the matrix $\widehat{Q}$ appearing in $\widehat\zeta_j$ and $\widetilde{Q}$ arising from the minimization step (since $\widehat\beta_j = \widetilde{Q}^{-1} \mathbb{U}_n\{ b(A_1)\widehat\zeta_j(Z_1, Z_2)\}$) will not be equal, even if they estimate the same matrix $Q = \Pb\{b(A)b(A)^T\}$. In particular, $\sgn\left\{b(a)^TQ^{-1}b(A_1)\right\}$ might not equal $\sgn\left\{b(a)^T\widetilde{Q}^{-1}b(A_1)\right\}$ and so $\widehat{g}_\ell(a)$ could be larger than $\widehat{g}_u(a)$. However, this is expected to occur with vanishing probability as the sample size increases.

\begin{proposition}\label{prop:outcome_model}
Assume that:
\begin{enumerate}
	\item $e^TQ^{-1} h(A)$ has a bounded density with respect to the Lebesgue measure;
	\item $\left\| \int S_2\{\widehat\zeta_j(Z_1, z_2) - \zeta_j(Z_1, z_2)\} d\Pb(z_2) \right\| = o_\Pb(1)$;
	\item $\Norm{w - \widehat{w}} \Norm{\mu - \widehat{\mu}} = o_\Pb(n^{-1/2})$.
\end{enumerate}
Then $\sqrt{n}\{\widehat{g}_j(a) - g_j(a)\} \indist N\left(0, 4\Sigma \right)$, for $ j \in \{ u, \ell\},$ where
\begin{align*}
	\Sigma = b(a)^T\var\left[Q^{-1}\int S_2b(A_1)\{\zeta_j(Z_1, z_2) - b^T(A_1)\beta_j\}d\Pb(z_2)\right]b(a).
\end{align*}
\end{proposition}

Bounds on a specific coordinate of $\beta$, say $\beta_1$, are
straightforward to derive in the linear MSM case by replacing $b(a)^T$
with $e^T$ in the bounds above. When $g(a; \beta)$ is not linear,
bounds on $\beta_1$ can be obtained using a homotopy algorithm similar
to that in Section~\ref{section::homotopy}.  The algorithm uses the
functional derivative of $\beta(\Delta)$ with respect to $\Delta$ in
$L_2(\Pb(a))$:
\begin{align*}
\frac{\delta \beta(\Delta)}{\delta \Delta} = \E\left\{ h(A) \nabla_\beta g(A; \beta)^T \right\} ^{-1} h(A).
\end{align*}
Another, exact but computationally expensive, approach is described in Appendix~\ref{section::nonlinear}.

\section{Time Series}\label{section::time}

Now we extend the methods to time varying treatments.
In this setting, we have data 
$(X_1,A_1),\ldots,\allowbreak (X_T,A_T,Y)$ on each subject,
where $X_t$ can include an intermediate outcome
$Y_t$. We write $\overline{X}_t = (X_1, \ldots, X_t)$ and $\overline{A}_t= (A_1, \ldots, A_t)$. An intervention corresponds to setting
$\overline{A}_T = \overline{a}_{T} = (a_1,\ldots,a_T)$ with
corresponding counterfactual outcome $Y(\overline{a}_T)$. In this case, 
the assumption of no unmeasured confounding
is expressed as
$A_t \ind Y(\overline{a}_T) | (\overline{A}_{t-1}, \overline{X}_t)$
for every $t \in \{1, \ldots, T\}$.
Under this assumption, the $g$-formula  \citep{robins1986new} is 
$$ 
\E\{Y(\overline{a}_T)\} = \int \mu(\overline{a}_T,\overline{x}_T)
\prod_{s=1}^T d\Pb(x_s|\overline{x}_{s-1},\overline{a}_{s-1}) 
$$
where $\mu(\overline{a}_T,\overline{x}_T) = \E(Y \mid  \overline{X}_T = \overline{x}_T,\overline{A}_T = \overline{a}_T)$. 
As before, a MSM is a model $g(\overline{a}_T;\beta)$ for
$\E\{Y(\overline{a}_T)\}$. A common example is
$g(\overline{a}_T;\beta) = \beta_0 + \beta_1 \sum_{s=1}^T a_s$. For some user-specified function $h(\cdot)$ of the treatments, it can be shown that
\begin{align*}
	\E\left[ h(\overline{A}_T) W_T(\overline{A}_T, \overline{X}_T) \left\{ Y - g(\overline{A}_T; \beta) \right\} \right] = 0, \quad \text{where} \quad	W_T (\overline{a}_T, \overline{x}_T)=  \frac{\prod_{s = 1}^T \pi(a_s |\overline{a}_{s - 1})}{\prod_{s = 1}^T
		\pi(a_s| \overline{x}_{s},  \overline{a}_{s -1})}.
\end{align*}

\subsection{Bounds on $g(\overline{a}_t;\beta)$ under Propensity Sensitivity Confounding}
Let $\overline{U}_T = (U_1, \ldots, U_T)$ denote unobserved
confounders. If $A_t \ind Y_t(\overline{a}_t) | (\overline{A}_{t-1}, \overline{X}_t, \overline{U}_t)$ for all $t$, then the $g$-formula becomes
\begin{align*}
\E\{Y(\overline{a}_T)\} = 
\int \E(Y \mid \overline{A}_T = \overline{a}_T, \overline{X}_T = \overline{x}_T, \overline{U}_T = \overline{u}_T)
\prod_{s=1}^T d\Pb(x_s, u_s | \overline{x}_{s-1}, \overline{u}_{s - 1}, \overline{a}_{s-1}),
\end{align*}
Define
\begin{align*}
v_T(Y, \overline{A}_T, \overline{X}_T) =
	\E\left\{ \frac{\prod_{s = 1}^T\pi(A_s | \overline{X}_{s}, \overline{A}_{s - 1})}
	{\prod_{s = 1}^T\pi(A_s | \overline{X}_{s}, \overline{U}_{s}, \overline{A}_{s -1})} \mid Y, \overline{A}_T, \overline{X}_T\right\}
\end{align*}
and note that we can rewrite $\E\{Y(\overline{a}_T)\}$ as
\begin{align*}
\E\{Y(\overline{a}_T)\} & 
= \int \E\{Y v_T(Y, \overline{a}_T, \overline{x}_T) |\overline{A}_T = \overline{a}_T, \overline{X}_T = \overline{x}_T\} \prod_{s=1}^T d\Pb(x_s|\overline{x}_{s-1},\overline{a}_{s-1}).
\end{align*}
It can be shown that $\E\{v_T(Y, \overline{A}_T, \overline{X}_T)\} = 1$ and also that
\begin{align}\label{eq:time_var_cond_rest}
	\int \E\{v_T(Y, \overline{A}_T, \overline{X}_T) \mid \overline{A}_T, \overline{X}_T\} \prod_{s = 2}^T d\Pb(x_s \mid \overline{x}_{s-1}, \overline{a}_{s-1}) = 1
\end{align}
However, unless additional assumptions are invoked, it is not the case that $\E\{v_T(Y, \overline{A}_T, \overline{X}_T) \mid \overline{A}_T, \overline{X}_T\} = 1$. Getting bounds in the propensity sensitivity model enforcing $v_T(Y, \overline{A}_T, \overline{X}_T) \in [\gamma^{-1}, \gamma]$ and $\E\{v_T(Y, \overline{A}_T, \overline{X}_T)\} = 1$ is straightforward. For example, as shown in Section \ref{appendix:mc_time_varying} in the appendix, it holds that
\begin{align*}
	\E\left[ h(\overline{A}_T) W_T(\overline{A}_T, \overline{X}_T) \left\{ Yv_T(Y, \overline{A}_T, \overline{X}_T) - g(\overline{A}_T; \beta) \right\} \right] = 0
\end{align*}
In this light, methods based on the class $\mathcal{V}_{\text{large}}(\gamma)$ described in Sections \ref{section::homotopy} and \ref{sec::betalin} apply here as well with $W_T$ replacing $W$ and $v_T$ replacing $v$. The local approach taken in Section \ref{sec::local} also applies. However, enforcing \eqref{eq:time_var_cond_rest} appears more challenging and we leave it for future work.
\subsection{Bounds under Outcome Sensitivity Confounding}

Bounds for $g(\overline{a}_T;\beta)$ and for coordinates of $\beta$ governed by the outcome 
sensitivity model can be derived in a similar fashion by extending the results in Section~\ref{section::MSMOS}.

\section{Examples}
\label{section::examples}

In this section we present a static treatment example and a time series example. 
The appendix also contains simple, proof of concept synthetic examples.

\subsection{Effect of Mothers' Smoking on Infant Birthweight}

\begin{figure}
\begin{center}
\includegraphics[scale=.75]{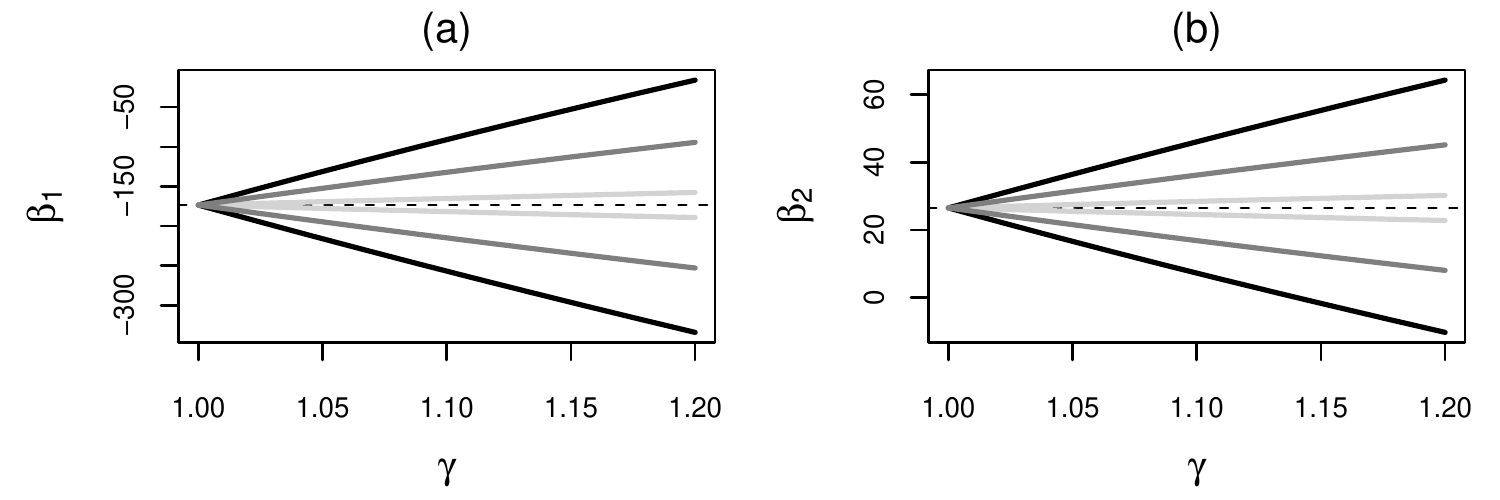}
\end{center}
\caption{\em {\bf Bounds for $\beta_1$ and $\beta_2$ for the
birthweight dataset, assuming the MSM $g(a;\beta) = \beta_0+\beta_1 a    + \beta_2 a^2$}.  The dotted horizontal lines are at $\hat \beta_1$ and 
$\hat \beta_2$. The black bounds are from $F_1$ over 
${\cal V}_{\rm large}$ for the propensity model, found using the homotopy
algorithm (Section~\ref{section::homotopy}).
The local approximation to $F_1$ (Section~\ref{sec::local}) matched the black bounds closely (not shown), similar to the appendix examples.
The quadratic term parameter loses significance at
$\gamma \approx 1.11$ and the linear term at $\gamma \approx 1.25$.
The dark and light grey bounds use the subset sensitivity model (Sections~\ref{section::subset} and~\ref{section::subset2}), 
with
$\epsilon=.5$ and~$.1$, respectively. Bounds are all the narrower
when $\epsilon$ is smaller, as expected.}
\label{fig::smokingeps}
\end{figure}

We re-analyzed a dataset of births in Pennsylvania
between 1989 and 1991, which has been used to investigate the
causal effects of mothers' smoking behavior on infants
birthweight. Previous analyses \citep{almond2005costs,  cattaneo2010efficient}, assuming no unmeasured confounders, found a
negative effect of smoking on the infant's weight. Recently, \cite{ed}
conducted a sensitivity analysis to the assumption of no unmeasured
confounding by dichotomizing the treatment into smoking vs
non-smoking. In line with previous work, they found a negative effect
of smoking on the child's weight, but also identified
plausible values of their sensitivity parameters consistent with a
null effect. They concluded that, while likely negative, the true
effect of smoking on weight might be smaller than that estimated under
no unmeasured confounding. We complement and expand on their analysis
by considering sensitivity models that can accommodate MSMs; we reach
similar conclusions, although we find the estimated effect to be less
sensitive to the unmeasured confounding parametrized by our
sensitivity models.

The dataset consists of a random subsample of $5,000$ observations
from the original dataset that is available
online.\footnote{\url{https://github.com/mdcattaneo/replication-C_2010_JOE}}
The outcome is birthweight and the treatment is an ordered categorical
variable taking six values corresponding to ranges $\{0, 1\text{-}5,
6\text{-}10, 11\text{-}15, 16\text{-}20, 21+ \}$ of cigarettes smoked
per day. There are 53 pre-treatment covariates including mother's and
father's education, race, and age; mother's marital status and foreign
born status; indicators for trimester of first prenatal care visit and
mother's alcohol use.

\begin{figure}[t!]
\begin{subfigure}{.45\textwidth}
\centering
\includegraphics[width=.9\linewidth]{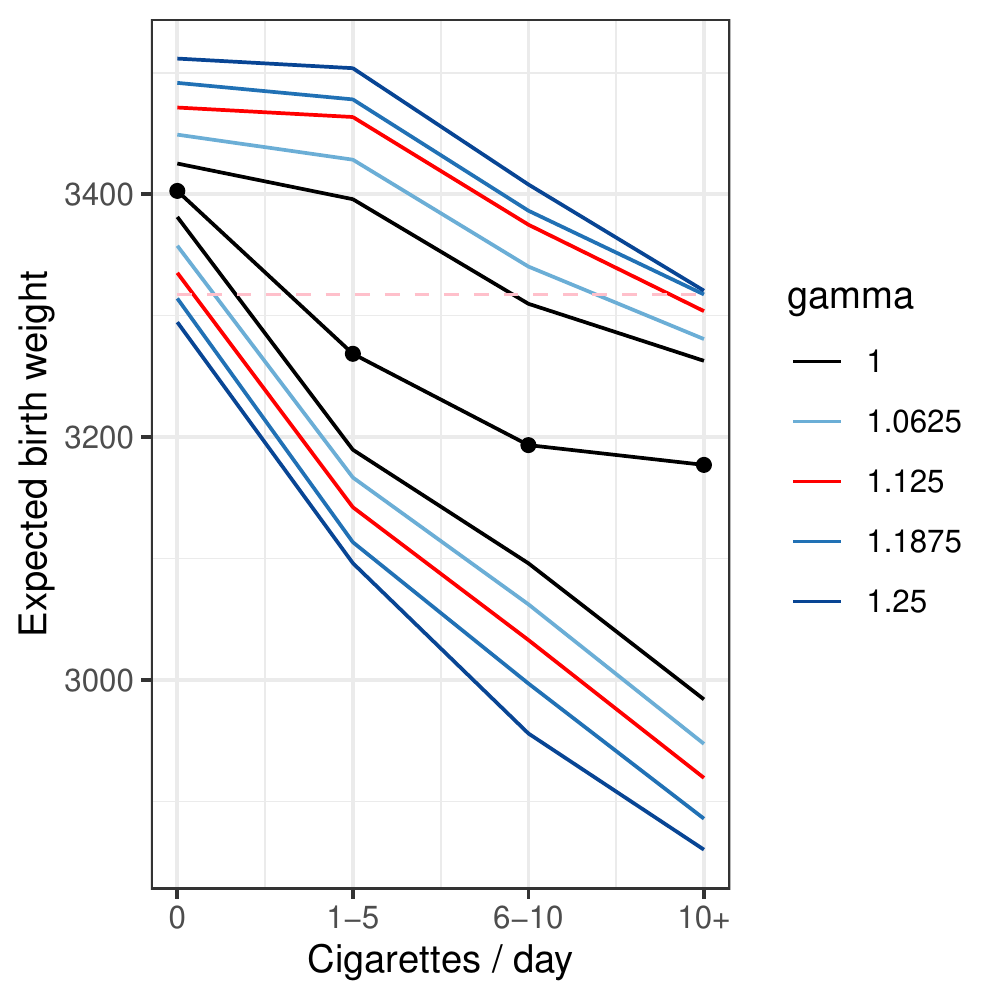}
\caption{\em Estimated bands for
the linear MSM $g(a; \beta) = \beta_0 + \beta_1 a + \beta_2a^2$.}
\label{fig:sfig1}
\end{subfigure} \hfill
\begin{subfigure}{.45\textwidth}
\centering
\includegraphics[width=.9\linewidth]{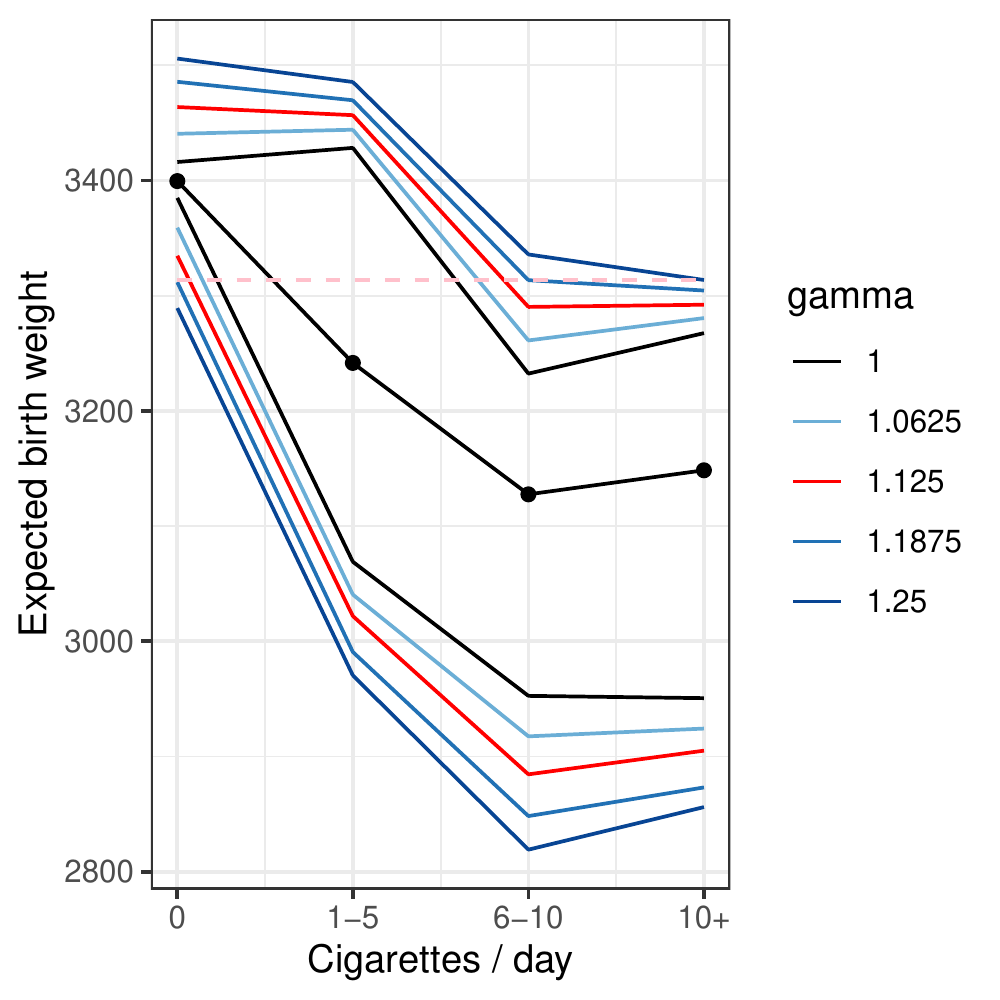}
\caption{\em Estimated bands for the saturated MSM
 $g(a; \beta) = \beta_0 + \sum_{j = 1}^4 \beta_j \one\{a_j \in j^{\text{th}} \text{bin} \}$.}
\label{fig:sfig2}
\end{subfigure}
\caption{\em {\bf Pointwise 95\%-confidence bands on the bounds for $\E \{ Y(a) \} = g(a; \beta)$ under the propensity sensitivity model}, where $a \in \{0,1$--$5, 6$--$10, 10+\}$ cigarettes per day. The lines with dots are $g(a; \hat \beta)$. }
\label{fig:smoking_msm_g}
\end{figure}

Figure~\ref{fig::smokingeps} shows bounds on $\beta_ 1$ and $\beta_2$ for the 
quadratic MSM given by $g(a;\beta) = \beta_0+\beta_1 a + \beta_2 a^2$ under propensity sensitivity, based on $F_1$ 
over ${\cal V}_{\rm large}$ (with only six treatment values, we cannot fit a more complex parametric model). 
We estimated the propensity $\pi(a|x)$ via a log-linear neural net using the
\texttt{nnet} package for the \texttt{R} software, as in \cite{cattaneo2010efficient}.
The quadratic term parameter loses significance at $\gamma \approx 1.11$ and the linear term at $\gamma=1.25$.
Figure~\ref{fig::smokingeps} also shows bounds on $\beta_ 1$ and $\beta_2$
under the subset sensitivity model with $\epsilon=.5$ and~$.1$.
As expected, there is much less sensitivity for small $\epsilon$.

Recall from (\ref{eq::Pi}) that $\gamma$ measures the change in the propensity score
when $U$ is dropped.
To determine if $\gamma=1.25$ constitutes substantial confounding, we followed the ideas in
\cite{cinelli2020making} by assessing changes 
to the propensity score when observed confounders are dropped.
Most authors drop one covariate at a time but with 53 covariates, we found that
this caused almost no changes to the propensity score.
Instead, we (i) dropped half of the covariates, and (ii) computed, for each data point, the ratio of 
propensity scores 
using all the covariates and the randomly chosen subset, and repeated 
(i, ii) 100 times. Each repeat yielded a distribution of propensity score ratios and
we used the average of their 80th percentiles as a measure of substantial
confounding. This value is $\gamma=1.20$, so we conclude that the causal effect of smoking on
infant birthweight remains significant even under substantial confounding. The next analysis confirms this conclusion.

Next, Figure~\ref{fig:smoking_msm_g} shows
$95\%$ point-wise confidence bands for the bounds on $g(a; \beta)$
under propensity sensitivity based on
$\mathcal{V}_\text{small}$, assuming that the bounds are modeled as $g(a; \beta_\ell)$ and $g(a; \beta_u)$; see Proposition \ref{prop:vdv}. 
Note that Figure~\ref{fig::smokingeps} showed the bounds on $\beta_1$ and $\beta_2$ rather than
confidence bands on these bounds, because confidence bands are difficult to obtain; see Sections~\ref{section::homotopy} and~\ref{sec::betalin}.
Figure~\ref{fig:sfig1} shows results for the quadratic MSM $g(a; \beta) = \beta_0 +
\beta_1 a + \beta_2 a^2$, and as
a safeguard against MSM mis-specification, Figure~\ref{fig:sfig2} shows the saturated parametric MSM fit.
The black bands corresponding to $\gamma = 1$ assume 
no-unmeasured-confounding (so they are confidence bands for $g(a; \hat \beta)$) and increasing values of $\gamma$ correspond to
increasing amount of unmeasured confounding. 
We estimated the nuisance functions nonparametrically:
the outcome
model $\mu(x,a)$ and conditional quantiles
$q_j(Y  |  a, x)$ were fitted assuming generalized additive models, with mother's and father's ages, education and birth order entering the model linearly, and number of prenatal care visits and months since last birth entering the model as smooth terms -- we used the \texttt{mgcv} and \texttt{qgam} packages in \texttt{R},
respectively; the propensity $\pi(a|x)$ was estimated via a log-linear
neural net, as above.
We constructed the $95\%$ point-wise confidence bands
relying on Proposition \ref{prop:vdv} and the Hulc method by
\cite{kuchibhotla2021hulc}. For the Hulc, the sample
needs to be split into six subsamples, but because of small sample sizes in some categories, we collapsed all regimes
of 10+ cigarettes into one category, thereby reducing the number of
treatment regimes to four. 
Consistent with Figure~\ref{fig::smokingeps} and previous analyses \citep{almond2005costs,
cattaneo2010efficient}, we found a statistically significant
negative relationship between smoking and birthweight under
no-unmeasured-confounding. The relation ceases to be significant for
$\gamma = 1.1875$ when the quadratic model is used and $\gamma = 1.25$ when the saturated model is used.

\subsection{Effect of Mobility on Covid-19 Deaths}

\begin{figure}[t!]
\hspace{-.15in} \includegraphics[width=6in,height=1.5in]{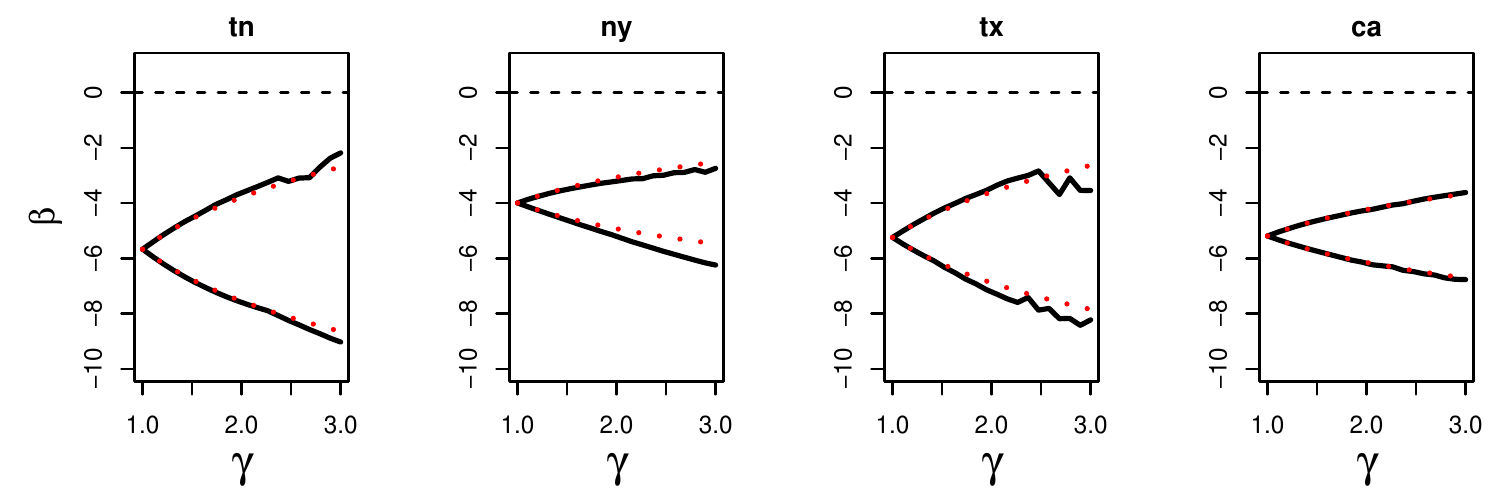}
\caption{{\bf Bounds on $\beta$ in MSM (\ref{eq::theMSM}) for the
Covid data in four US states}. \em The black bounds are from $F_1$
over ${\cal V}_{\rm large}$ with increasing amount of unobserved 
confounding, under the propensity confounding model. 
The occasional lack of smoothness is due to
estimating the quantile $q$ from small samples ($n=40$) in the
homotophy algorithm (Section~\ref{app::homo}).  The dotted red bounds are the local
approximations to $F_1$ (Section~\ref{sec::local}). The MSM coefficients remain significant under substantial unobserved confounding for the four
states: mobility has a significant effect on Covid deaths.}
\label{fig::beta}
\end{figure}

\begin{figure}[t!]
\begin{center}
\includegraphics[width=5in]{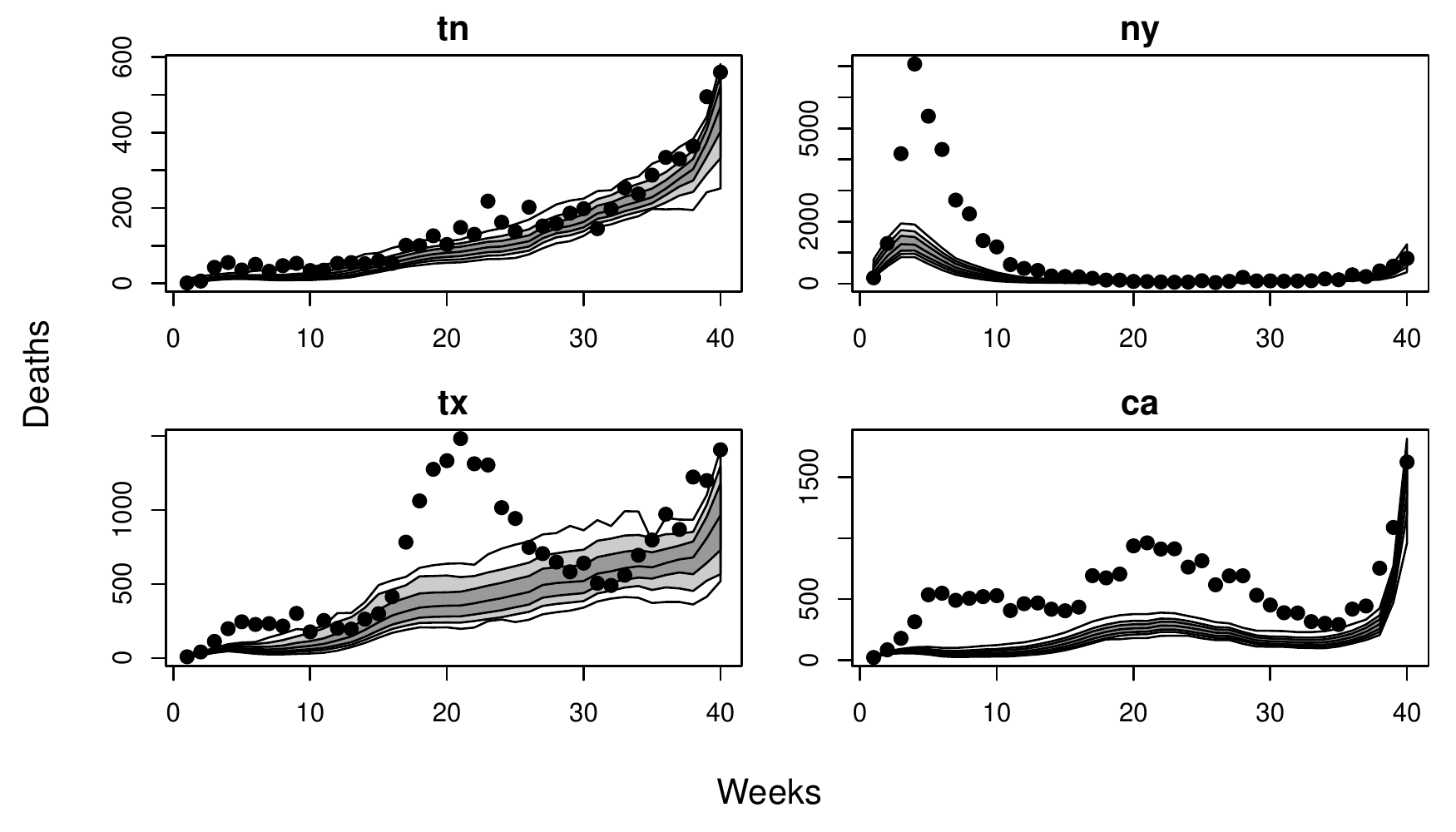}
\end{center}
\caption{\em {\bf Bounds for counterfactual 
deaths $\psi(a_T) = \E(Y(a_T))$ for the Covid data in four US states using MSM~(\ref{eq::theMSM}) in a  
hypothetical mobility scenario $a_T$
corresponding to shifting the observed mobility pattern two weeks earlier.}
The bounds are from $F_1$ over ${\cal V}_{\rm large}$ under propensity sensitivity, 
found using the homotopy algorithm. 
The shades correspond to $\gamma = 3$ (white), $\gamma = 2$ (light grey) and $\gamma = 1$ (no unmeasured confounding, dark grey).
The black dots are observed deaths. Our analysis suggests that, even with substantial unobserved confounding, sheltering two weeks earlier would have saved lives, although only by a small number in TN, because the epidemic there started more mildly.}
\label{fig::deaths}
\end{figure}

We revisit the analysis in
\cite{bonvini2021causal} on the causal effects of mobility on deaths
due to Covid-19 in the United States. In their paper, a sensitivity analysis to the
no unmeasured confounding assumption was conducted under the
propensity model without providing details. We provide details here.

The data consist of weekly observations, at the state level, on the
number of Covid-19 deaths $Y_t$ and a measure of mobility ``proportion
at home," $A_t$, which is the fraction of mobile devices that did not
leave the immediate area of their home. The time period considered in
the analysis is February 15 2020 (week 1) to November 15 2020 (week 40). 
We focus on four states, CA, FL, NY and TN, as 
representatives of four different evolutions of the pandemic; 
their observed time series of deaths are plotted as dots 
in Figure~\ref{fig::deaths}.
We model each state separately so that differences between states do not act as confounders of the
treatment/outcome relationship.

Our MSM is given by 
\begin{equation}
\label{eq::theMSM}
g(\overline{a}_t, \beta, \nu)= \mathbb{E}[L_t({\overline{a}_t})] = \nu(t)  + \beta M_t
\end{equation}
where $\overline{a}_t =(a_1,\ldots, a_t)$, $L_t(\overline{a}_t)$ are
log-counterfactual deaths, $L_t = \log(Y_t+1)$, $M_t \equiv
M(\overline{a}_t)= \sum_{s=1}^{t-\delta} a_s$, and $\delta = 4$ weeks
is approximately the mean time from infection to death from Covid-19.
The nuisance function $\nu(t)$ is assumed to be non-linear to capture
changes in death incidence due to time varying variables other than
mobility, for example probability of dying, which decreased over time
due to better hospital treatment, number of susceptibles to Covid-19,
which naturally decreased, and social distancing changes.

Figure~\ref{fig::beta} shows $\hat\beta$ for the four states, along
with lower and upper bounds under propensity sensitivity. The estimates are negative,
as would be expected since higher $A_s$ means that more people sheltered at home,
and they remain negative even under substantial unobserved confounding.

\cite{bonvini2021causal} also estimated counterfactual deaths under
three hypothetical mobility regimes $\overline{A}_t = (A_1,\ldots,
A_t)$: ``{\it start one week earlier}" and ``{\it start two weeks
earlier}", which shifts the observed mobility profiles back by one
or two weeks with aim to assess Covid-19 infections if
we had started sheltering in place one and two weeks earlier; and ``{\it stay vigilant}", which 
halves the slope of the rapid decrease in stay at home mobility after the initial peak in week 9, 
when a large proportion of the population hunkered down after witnessing the situation in 
New York city.
To save space, Figure~\ref{fig::deaths} shows only the estimated counterfactual deaths and bounds for the ``start two weeks earlier'' scenario.
Bounds were computed on $g(\overline{a}_t;\beta)$ for each $t$
using the homotopy algorithm on $F_1$ over ${\cal V}_{\rm large}$, under propensity sensitivity (Section~\ref{section::homotopy}).

Bounds for $\beta$ and $g(\overline{a}_t;\beta)$ under the outcome sensitivity model requires an 
outcome model, which we do not pursue here.

\section{Conclusion}
\label{section::conclusion}

We have derived several sensitivity analysis methods
for marginal structural models.
Doing so may require additional modeling,
for example, using quantile regression.
We also saw that approximate, conservative bounds are possible
without quantile regression.

We have focused on
the traditional interventions corresponding to
setting the treatment to a particular value.
In a future paper,
we address sensitivity analysis under stochastic interventions.
Here we find that these interventions can lead to 
inference that is less sensitive to unmeasured confounding than
traditional interventions. 

One issue that always arises
in sensitivity analysis
is how to systematically choose ranges of values
for the sensitivity parameters (e.g., $\gamma$, $\delta$, $\epsilon$, in our case). 
In the smoking example, we dropped large sets of
observed confounders to provide a benchmark, 
but for the most part this is an open problem.

\section{Acknowledgements}

We thank Prof. Nicole Pashley for helpful discussions regarding the
interpretation of the causal effect of mobility on deaths due to
Covid-19. In particular, unmeasured confounding is not the only issue
that needs to be addressed when interpreting our results. A potential
complication is that there could be multiple versions of mobility,
e.g. a person may move to go to work versus a bar. These different
versions of mobility may affect the probability of dying due to
Covid-19 differently, complicating the interpretation of the overall
effect of reduced mobility on deaths. Conducting a sensitivity
analysis to gauge the impact of multiple versions of the same
treatment is an important avenue for future work.

\bibliographystyle{plainnat}
\bibliography{paper.bib}

\newpage

\appendix

\section{Appendix}

\subsection{Synthetic Examples \label{app:examples}}

As a proof of concept,
we consider a simple simulated example.
We take $n=100$,
$A_1,\ldots, A_n\sim N(0,1)$
and $Y_i = \beta A_i + \epsilon_i$ where
$\epsilon_i \sim N(0,1)$.
Figure~\ref{fig::simulated}a
shows the propensity sensitivity bounds for $\beta=3$ based on the homotopy algorithm and using the local
approximation; the local method is an excellent approximation.
We also conducted a few simulations using very small sample sizes where
the exact solution can be computed by brute force.
We found that the homotopy method was indistinguishable from the exact bound.
Figure~\ref{fig::simulated}b shows the bounds using the outcome sensitivity approach.

Now we look at the effect
of bounding 
over ${\cal V}_{\rm large}$
using $F_1$ and $F_2$.
The example in Figure \ref{fig::explain}
shows that the propensity sensitivity 
bounds from $F_2$ (green) are wider than
bounds from $F_1$ (black).
In this case we used $n=1000$,1
$X\sim N(0,1)$,
$A = X + N(0,1)$,
$Y = \beta A + 2X + N(0,1)$,
with $\beta = 3$.
Conversely, the example in  Figure \ref{fig::explain}b
shows that the propensity sensitivity 
bounds from $F_2$ (green) are narrower than
bounds from $F_1$ (black).
Here we used $n=1000$ with:
$U\sim {\rm Unif}(.5,1)$,
$A = 3-U$,
$Y = 5U$ and
$Y = 2.5 U + .25 N(0,1)$.
The red dotted lines are the local approximations to the $F_1$ bounds, which
are very good in these two examples as well.
Our experience is that usually $F_1$ gives tighter bounds.

\begin{figure}[t!]
\begin{center}
\includegraphics[scale=0.7]{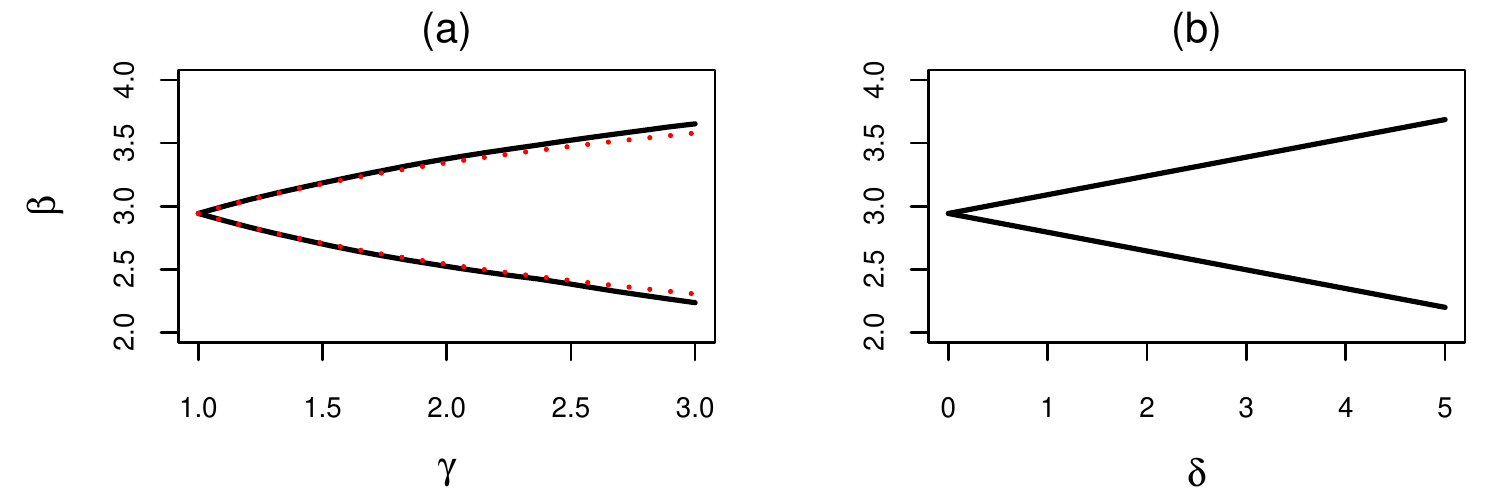}
\end{center}
\caption{\em {\bf (a) Propensity
sensitivity model bounds and (b) Outcome sensitivity model bounds for $\beta$ 
in the MSM $g(a; \beta) = \beta a$ with $\beta=3$},
\em for a simulated example. In (a), the black bounds in are from $F_1$
over ${\cal V}_{\rm large}$, obtained by the homotopy algorithm (Section~\ref{section::homotopy}), and
their local approximations (Section~\ref{sec::local}) are in dotted red.}
\label{fig::simulated}
\begin{center}
\includegraphics[scale=.7]{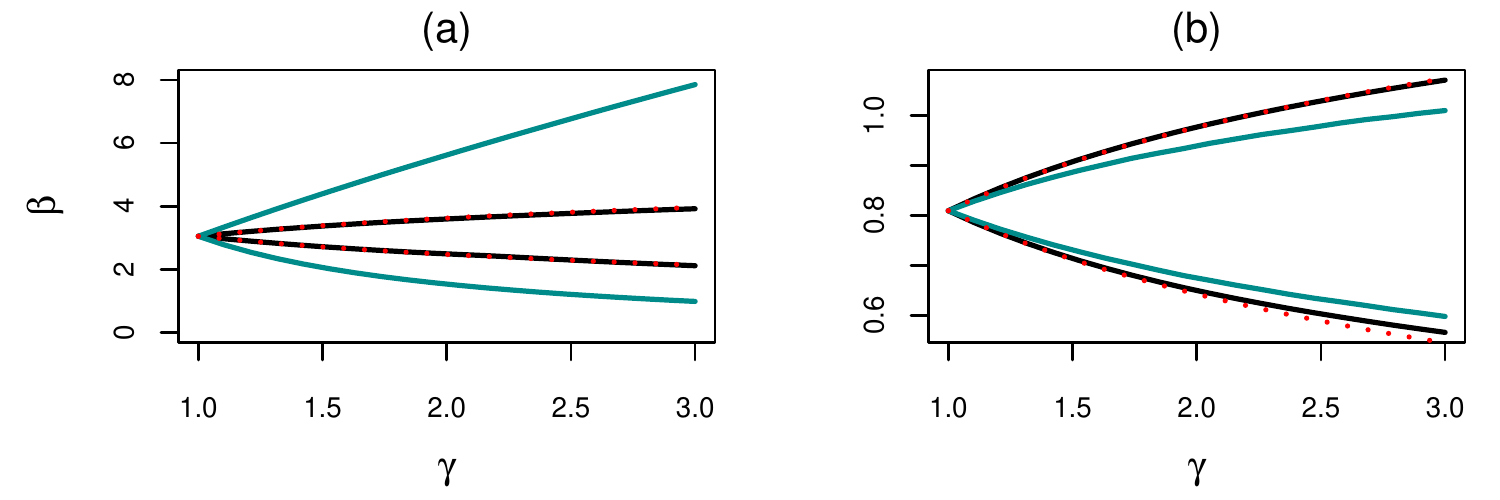}
\end{center}
\caption{
\em {\bf Propensity sensitivity model bounds from $F_1$ (black) and $F_2$ (green) 
over ${\cal V}_{\rm large}$}, \em in simulated examples. Bounds from $F_2$ are 
computationally easier to obtain than bounds from $F_1$. The local approximations to $F_1$, which 
are also simple to obtain, are in dotted red. (a) In this example, bounds 
from $F_2$ are wider than bounds from $F_1$, so we use the latter. (b) In this other example, the reverse is true. We do not know in advance which bounds to use.}
\label{fig::explain}
\end{figure}

\subsection{Subset Confounding}
\label{section::subset2}

Recall that, under this model,
an unknown proportion of the population
is subject to unobserved confounding.
Suppose $S$ is such that $Y(a) \ind A  |  X, S = 1$ but $Y(a) \ind A  |  X, U, S = 0$ where $U$ is not observed. That is, $S=0$ represents
the subset with unmeasured confounding and $S=1$ represents the subset
with no unmeasured confounding. This is a sensitivity model proposed
by \cite{bonvini2021sensitivity} in the case of binary treatments. Here, we extend this
framework to multivalued treatments and MSMs under the
propensity sensitivity model and the outcome sensitivity
model. 
To start, we define the propensity model in this
case to be
\begin{align*}
\gamma^{-1} \leq \frac{\pi(a | x, u, S = 0)}{\pi(a | x, S = 0)} \leq \gamma \quad \text{ for all } a, x, u.
\end{align*}
Let
\begin{align*}
v_0(Z) = \E\left\{\frac{\pi(A  |  X, S = 0)}{\pi(A  |  X, S = 0, U)} \ \Big| \ Y, A, X, S = 0 \right\}
\end{align*}
and notice that $\E\{v_0(Z)  |  A, X, S = 0\} = 1$ and $v_0(Z) \in
[\gamma^{-1}, \gamma]$. Essentially, we can repeat the same
calculations as in the non-contaminated 
model, this time simply
applied to the $S = 0$ group.

The same argument used in proving Lemma \ref{lemma:np_cond_bounds} yields that
\begin{align*}
	\E\{Y\gamma^{\sgn\{q_\ell (Y | a, x, S = 0) - Y\}}  |  A = a, X = x, S = 0\}\\
	\leq \E\{Y v_0(Z)  |  A = a, X = x, S = 0\} \leq\\
	\E\{Y \gamma^{\sgn\{Y - q_u(Y | a, x, S = 0)\}}  |  A = a, X = x, S = 0\}
\end{align*}
where $q_\ell (Y | a, x, S = 0)$ and $q_u(Y | a, x, S = 0)$ are the 
$\tau_\ell = 1 / (\gamma + 1)$ and $\tau_u = \gamma / ( 1 + \gamma)$ quantiles of the distribution of $Y  |  (A, X, S = 0)$. As in \cite{bonvini2021sensitivity}, we make the simplifying assumption that $S \ind Y  |  A, X$. This way, 
\begin{align*}
	m_\ell(a, x) \equiv \E\{Y\gamma^{\sgn\{q_\ell (Y | a, x) - Y\}}  |  A = a, X = x\}\\
	\leq \E\{Y v_0(Z) |  A = a, X = x, S = 0\} \leq\\
	\E\{Y \gamma^{\sgn\{Y - q_u(Y | a, x)\}} |  A = a, X = x\} \equiv m_u(a, x)
\end{align*}
where now the quantile are those of the distribution of $Y | (A, X)$. Notice that these are the usual bounds in the non-contaminated model.

We can then compute the bounds on $\E \{ Y(a) | X \}  $. First notice that,
\begin{align*}
\E \{ Y(a)  S |  X\} & = \E(Y |  A = a, X, S = 1) \Pb(S = 1  |  X) &  (Y(a) \ind A  |  X, S = 1) \\
& = \mu(a, X)\Pb(S = 1  |  X). & (Y \ind S  |  A, X)
\end{align*}
This means that $\E \{ Y(a)  S \} = \E\{S \mu(a, X)\}$. Next notice that
\begin{align*}
\E\{Y(a) (1 - S)  |  X\} & = \E \{ Y(a)   |  X, S = 0 \} \Pb(S = 0  |  X) \\
& = \E\{ Y \alpha(a, X, U, S = 0)  |  A = a, X, S = 0 \} \Pb(S = 0  |  X) \\
& = \E\{ Y v_0(Z) |  A = a, X, S = 0 \} \Pb(S = 0 |  X).
\end{align*}
Therefore, 
\begin{align*}
	\E \{ Y(a)   |  X \} = \mu(a, X) \Pb(S = 1  |  X) + \E\{ Y v_0(Z)  |  A = a, X, S = 0 \} \Pb(S = 0  |  X)
\end{align*}
so that
\begin{align*}
	& \E \{ Y(a) \} = \E\left\{\mu(a, X) \right\} + \E\left[ (1 - S) \left\{\E\{ Y v_0(Z)  |  A = a, X, S = 0 \} - \mu(a, X)\right\} \right]
\end{align*}
which implies, for $r_j(a, X) = m_j(a, X) - \mu(a, X)$ and $j \in \{l, u\}$:
\begin{align*}
	\E\{\mu(a, X)  + (1 - S) r_\ell(a, X)\} \leq \E \{ Y(a) \} \leq \E\{\mu(a, X)  + (1 - S) r_u(a, X)\}.
\end{align*}
Let $t_{\epsilon, l}(a)$ the $\epsilon$-quantile of $r_\ell(a, X)$ and
$t_{\epsilon, u}$ be the $(1-\epsilon)$-quantile of $r_u(a,
X)$. Further, let $\lambda_\ell(a, x) = \one\{ r_u(a, x) \leq
t_{\epsilon, l}(a) \}$ and $\lambda_u(a, x) = \one\{ r_u(a, x) >
t_{\epsilon, u}(a) \}$. Under the assumption that $\Pb(S = 0) =
\epsilon$, we bound $\E \{ Y(a) \}$ by further optimizing over $S$ as
\begin{align*}
	& \E \{ Y(a) \} \geq  \E\{\mu(a, X) \} + \E[\lambda_\ell(a, X)r_\ell(a, X)] \equiv \theta_\ell(a) \\
	& \E \{ Y(a) \} \leq  \E\{\mu(a, X) \} + \E[\lambda_u(a, X)r_u(a, X)] \equiv \theta_u(a).
\end{align*}
As discussed in Section \ref{sec:bounds_g}, one option to estimate the
bounds is to assume that they follow some parametric models $g(a;
\beta_\ell)$ and $g(a; \beta_u)$. Define
\begin{align*}
	& f_\mu(Z_1, Z_2) = W(A_1, X_1)\{Y_1 - \mu(A_1, X_1)\} + \mu(A_1, X_2) \\
	& f_{\Delta, j}(Z_1) = W(A_1, X_1)\left[\left\{s_j(Z; q_j) - \kappa(A_1, X_1; q_j) \right\} - Y_1 + \mu(A_1, X_1) \right] \\
	& f_{r, j}(Z_1, Z_2) = \kappa(A_1, X_2; q_j) - \mu(A_1, X_2) \\
	& f_j(Z_1, Z_2) = f_\mu(Z_1, Z_2) + \lambda_j(A_1, X_1) f_{\Delta, j}(Z_1) + \lambda_j(A_1, X_2)f_{r, j}(Z_1, Z_2).
\end{align*}
Then, if it is assumed that $\theta_j(a) = g(a; \beta_j)$, the following moment condition holds
\begin{align*}
	\mathbb{U}\left[h(A_1)\left\{ f_j(Z_1, Z_2) - g(A_1; \beta_j) \right\} \right] = 0.
\end{align*}
In this respect, we define $\widehat\beta_j$ to solve
$\mathbb{U}_n\left[h(A_1)\left\{ \widehat{f}_j(Z_1, Z_2) - g(A_1;\widehat\beta_j) \right\} \right] = 0$. We 
estimate the nuisance functions on a separate sample $D^n$
independent from the sample $Z^n$ used to evaluate the
$U$-statistic. However, in the proof of the proposition below, we
require that $\widehat{t}_{\epsilon, j}(a)$ satisfies
\begin{align*}
\frac{1}{n}\sum_{i \in Z^n} \one\{\widehat{r}_j(a, X_i) > \widehat{t}_{\epsilon, j}(a)\} = \epsilon + o_\Pb(n^{-1/2}) \text{ for all } a \in \mathcal{A} \text{ and } j = \{l, u\}.
\end{align*}
In other words, we estimate all nuisance functions on a separate,
training sample except for $a \mapsto t_{1-\epsilon}(a)$, which is
estimated on the same sample used to estimate the moment
condition. This helps with controlling the bias due to the presence
of the indicator at the expense of an additional requirement on the
complexity of the class where $a \mapsto \widehat{t}_{
\epsilon, j}(a)$ belongs to. We have the following proposition.
 
\begin{proposition}\label{prop::bound_g_epsilon}
Suppose
\begin{enumerate}
\item The function class $\mathcal{G}_l = \left\{ a \mapsto h_l(a)g(a; \beta)\right\}$ is Donsker for every $l = \{1, \ldots, k\}$ with integrable envelop and $g(a; \beta)$ is a continuous function of $\beta$;
\item The map $\beta \mapsto \mathbb{U}\{h(A_1)f_j(Z_1, Z_2) - g(A_1; \beta)\}$ is differentiable 
at all $\beta$
with continuosly invertible matrices $\dot\Psi_{\beta_0}$ and $\dot\Psi_{\widehat\beta}$, where $\dot\Psi_{\beta} = -\E\{h(A)\nabla^Tg(A; \beta)\}$;
\item The function class $\mathcal{T}$ where $a \mapsto \widehat{t}_{\epsilon, j}(a)$ and $a \mapsto t_{\epsilon, j}(a)$ belong to is VC-subgraph;
\item For any $a \in \mathcal{A}$ and 
$x \in \mathcal{X}$, $r_j(a, X) - t_{\epsilon, j}(a)$, $r_j(A, x) - t_{\epsilon, j}(A)$ and $r_j(A, X) - t_{\epsilon, j}(A)$ have bounded densities;
\item The following holds
\begin{align*}
			& \Norm{\widehat{q}_u - q_u} = o_\Pb(n^{-1/4}), \ \left\|r_j -\widehat{r}_j\right\|_\infty +  \left\| t_{\epsilon, j} - \widehat{t}_{\epsilon, j}\right\|_\infty = o_\Pb(n^{-1/4}) \\
			& \left(\left\|r_j -\widehat{r}_j\right\|_\infty + \left\| t_{\epsilon, j} - \widehat{t}_{\epsilon, j}\right\|_\infty + \Norm{\widehat{w} - w} \right)\left(\Norm{\widehat{\mu} - \mu} + \Norm{\kappa_j - \widehat\kappa_j}\right) = o_\Pb(n^{-1/2}).
\end{align*}
\end{enumerate}
Then,
\begin{align*}
	\sqrt{n}(\widehat\beta_j - \beta_j) \indist N(0,4 \Sigma)
\end{align*}
where
$	\Sigma = \var\left[ \dot\Psi^{-1}_{\beta_j} \int S_2h(A_1)\left\{ f_j(Z_1, z_2) - t_{\epsilon, j}(A_1)\lambda_j(A_1, x_2) - g(A_1; \beta_j)\right\} d\Pb(z_2) \right].$
\end{proposition}

To get bounds on some coordinate of $\beta$, say $\beta_1$, one may
proceed by homotopy as in the non-contaminated model. In the linear
MSM case, i.e. $g(a; \beta) = b(a)^T\beta$, bounds on $\beta_1$ that
enforce the restriction $\E\{v_0(Z)  |  A, X, S = 0\} = 1$ would be
\begin{align*}
l_\gamma = \int \min\left\{ e^T M^{-1} b(a) \theta_u(a),  e^T M^{-1} b(a) \theta_\ell(a) \right\} d\Pb(a) \\
u_\gamma = \int \max\left\{ e^T M^{-1} b(a) \theta_u(a),  e^T M^{-1} b(a) \theta_\ell(a) \right\} d\Pb(a)
\end{align*}
where $M = \E\{b(A)b(A)^T\}$ and we set $h(A) = b(A)$. A similar
statement to Proposition \ref{prop::glinear} can be derived using the
influence function established in proving Proposition
\ref{prop::bound_g_epsilon}.

{\bf Remark:}
{\em
If we make
the stronger assumption that
$S$ is independent of $(X,A,Y)$ then
$p_0(x,a,y) = p_1(x,a,y)$.
In this case it is easy to see that
$\E[ h(A) (Y - g(A,\beta)) w(A,X)  ((1-\epsilon) + \epsilon v(Z))] = 0$.
All the previous methods can then be used with $v$ replaced with 
$(1-\epsilon) + \epsilon v(Z)$.
}

Now we use the outcome sensitivity model
on the confounded subpopulation.
We will assume that $S\ind Z$.
The distribution is
$$
(1-\epsilon) p(u,x,a)p(y|x,a) + \epsilon p(u,x,a)p(y|u,x,a).
$$
The moment condition is
\begin{align*}
0 &= \int b(a) (y - b^T(a) \beta) w(u,x,a) dP(u,x,a) \\
&=
(1-\epsilon) \int b(a) (y - b^T(a) \beta) w(u,x,a) p(u,x,a) p(y|x,a) +
\epsilon \int b(a) (y - b^T(a) \beta) w(u,x,a) p(u,x,a) p(y|u,x,a)\\
&=
(1-\epsilon) \int b(a) (y - b^T(a) \beta) w(x,a) p(x,a) p(y|x,a) +
\epsilon \int b(a) (y - b^T(a) \beta) w(u,x,a) p(u,x,a) p(y|u,x,a)\\
&=
(1-\epsilon) \int b(a) (y - b^T(a) \beta) w(x,a) p(x,a) p(y|x,a) +
\epsilon \int b(a) (y - \mu(x,a)) w(u,x,a) p(u,x,a) p(y|u,x,a) \\
&\hspace{1cm} + \epsilon \int b(a) (\mu(x,a) - b^T(a) \beta) w(u,x,a) p(u,x,a) p(y|u,x,a)\\
&=
(1-\epsilon) \int b(a) (y - b^T(a) \beta) w(x,a) p(x,a) p(y|x,a) +
\epsilon \underbrace{\int b(a) (\mu(u,x,a) - \mu(x,a)) w(u,x,a) p(u,x,a)}_{\Xi}  \\
&\hspace{1cm} + \epsilon \int b(a) (\mu(x,a) - b^T(a) \beta) w(x,a) p(x,a)\\
&=
(1-\epsilon) \int b(a) (y - b^T(a) \beta) w(x,a) p(x,a) p(y|x,a) +
\epsilon \int b(a) (\mu(x,a) - b^T(a) \beta) w(x,a) p(x,a) p(y|x,a) + \epsilon \Xi\\
&=
\E\Biggl[ b(A) \Bigl( (1-\epsilon)Y + \epsilon \mu(X,A)\Bigr) w(X,A)\Biggr] - \underbrace{\Biggl(\E[ b(A)b^T(A) w(X,A)]\Biggr)}_{\Omega} \beta + 
\epsilon \Xi,
\end{align*}
where
$
\Xi = \int b(a) (\mu(u,x,a) - \mu(x,a)) w(u,x,a) p(u,x,a)
$
and
$
\Omega = \E[ b(A)b^T(A) w(X,A)].
$
Therefore
$$
\beta = \Omega^{-1} 
\E\Biggl[ b(A) \Bigl( (1-\epsilon)Y + \epsilon \mu(X,A)\Bigr) w(X,A)\Biggr] + \epsilon \Omega^{-1}\Xi
$$
and
$
\beta_1 = e^T \beta$,
where
$e = (1,0,\ldots, 0)$.
Let $r$ be the first row of $\Omega^{-1}$
and let
$f(a) = \sum_j r_j b_j(a)$.
Then
\begin{align*}
e^T \Omega^{-1}\Xi & = r^T \Xi = \int (\sum_j r_j  b_j(a)) (\mu(u,x,a) - \mu(x,a)) w(u,x,a) p(u,x,a)\\
& \leq \delta \int f(a) I(f(a)>0) \pi(a) - \delta \int f(a) I(f(a)<0) \pi(a)\\
&= \delta \int f(a) (2I(f(a)>0) -1)\pi(a).
\end{align*}
Similarly,
$$
e^T \Omega^{-1}\Xi \geq 
\delta \int f(a) I(f(a)<0) \pi(a) - \delta \int f(a) I(f(a)>0) \pi(a)\\
-\delta \int f(a) (2 I(f(a)>0)-1) \pi(a).
$$
Therefore,
$$
\beta_1^* - \delta \int f(a) (2 I(f(a)>0)-1) \pi(a) \leq \beta_1 \leq
\beta_1^* + \delta \int f(a) (2 I(f(a)>0)-1) \pi(a)
$$
where
$$
\beta_1^* = \E\Biggl[ b(A) \Bigl( (1-\epsilon)Y + \epsilon \mu(X,A)\Bigr) w(X,A)\Biggr].
$$

\subsection{Bounds for $\beta$ under the outcome sensitivity confounding model when the MSM is not linear} 
\label{section::nonlinear}

Say the MSM is not linear.
Since
$g(a;\beta)=\E \{ Y(a) \}=\int \int y p(y|u,x,a) dP(x,u)$, we have
\begin{align*}
0 &= \int\int\int  h(a) (y-g(a;\beta)) p(y|u,x,a) \pi(a)dy dP(u,x)\\
&= \int\int h(a) (\mu(u,x,a)-g(a;\beta)) \pi(a) dP(u,x)\\
&= \int\int h(a) (\mu(u,x,a)- \mu(x,a))  \pi(a) dP(u,x) + \int\int h(a) (\mu(x,a)-g(a;\beta))\pi(a)  dP(u,x)\\
&= \int\int h(a) (\mu(u,x,a)- \mu(x,a))  \pi(a) dP(u,x) + \int\int h(a) (\mu(x,a)-g(a;\beta))\pi(a)  dP(x)\\
&= \int\int h(a) (\mu(u,x,a)- \mu(x,a))  \pi(a) dP(u,x) + \int\int \frac{h(a) (\mu(x,a)-g(a;\beta))\pi(a)}{\pi(a|x)}\pi(a|x)  dP(x)\\
&= \int h(a) \xi(a) \pi(a) da + \E[ h(A) (\mu(X,A)-g(A;\beta)) w(A,X)].
\end{align*}
Let
$C = \{ \E[h_1(A)\xi(A)],\ldots, \E[h_k(A)\xi(A)]:\ -\delta \leq \xi(a) \leq \delta\}$.
For each vector $t\in C$,
let
$\beta(t)$ solve
$\E[ h(A) (\mu(X,A)-g(A;\beta)) w(A,X)] = t$.
Then
$$
\inf_{t\in C} e^T \beta(t) \leq \beta_j \leq \sup_{t\in C} e^T \beta(t).
$$
These bounds can be found numerically by solving for
$\beta(t)$ over a grid on $C$.

\section{Algorithms}

\subsection{Homotopy Algorithm \label{app::homo}}

Input: grid $\{\gamma_1,\ldots,\gamma_N\}$ where
$\gamma_1 = 1$ and
$\gamma_1 < \cdots < \gamma_N$.
\begin{enumerate}
\item Let $\hat\beta$ be the solution of
$\sum_i  h(A_i)(Y_i - g(A_i;\hat\beta)) \hat W_i =0$.
Let $u_1 = \ell_1 = e^T \hat\beta$.
\item For $j=2,\ldots, N$:
\begin{enumerate}
\item Let 
$d_{j,i} \equiv d_{\gamma_{j-1},i}$ from~(\ref{eq::di1})
evaluated at $v = v_{\gamma_j -1}$.
\item
Set
$V_i = \gamma_j^{-1}I(d_{j,i} \leq q) + \gamma_j I(d_{j,i}>q)$
where
$q$ is the $\gamma_j/(1+\gamma_j)$ quantile of $d_{j,1},\ldots, d_{j,n}$.
Let $\hat\beta$ be the solution of
$\sum_i  h(A_i)(Y_i - g(A_i;\hat\beta)) \hat W_i V_i=0$. Set $u_j = e^T \hat\beta$.
\item Set
$V_i = \gamma_jI(d_{j,i} \leq q) + \gamma_j^{-1} I(d_{j,i} > q)$
where
$q$ is the $1/(1+\gamma_j)$ quantile of $d_{j,1} ,\ldots, d_{j,n}$.
Let $\hat\beta$ be the solution of
$\sum_i h(A_i)(Y_i - g(A_i;\hat\beta))\hat W_i V_i =0$. Set $\ell_j = e^T \hat\beta$.
\end{enumerate}
\item Return $(\ell_1,u_1),\ldots, (\ell_N,u_N)$.
\end{enumerate}

\subsection{Bounds on $\beta$ by Coordinate Ascent \label{app:coord}}

Another approach we consider
is coordinate ascent
where we maximize (or minimize)
$\hat\beta_1(v)$ over each
coordinate $v_i$ in turn.
It turns out that
this is quite easy
since 
$\hat\beta_1(v)$ 
is strictly monotonic in each $v_i$
for many models
so we need only compare the estimate
at the two values $v_i = \gamma$ and
$v_i = 1/\gamma$.
Furthermore,
in the linear case,
getting the estimate after
changing one 
coordinate $v_i$
can be done quickly using
a Sherman-Morrison rank one update.

The coordinate ascent approach
will lead to a local optimum but it
will depend on the ordering of the data
so we repeat the algorithm using several random orderings.
The homotopy method instead
uses the last solution as a starting point for the new solution.
This makes the 
homotopy method faster but, in principle, 
the coordinate ascent approach could explore
a wider set of possible solutions.
For simplicity, the only restriction we enforce is
$1/\gamma \leq v_i \leq \gamma$.
In practice, we find that the solutions are very similar.

\begin{lemma}
Suppose that the function $\hat\beta(v)$ is strictly
monotonic in each coordinate $v_i$.
The maximizer and minimizer occur at corners of the cube
$[1/\gamma,\gamma]^n$.
We have that
$$
\frac{\partial\hat\beta_1(v)}{\partial v_j}=
\frac{1}{W_i}
e^T \Bigl\{(X^T \mathbb{W} X)^{-1} [S_i - (R_i R_i^T) \hat\beta ]\Bigr\}
$$
where $W_i = 1/\pi(A_i |  X_i)$, 
$\mathbb{W}$ is diagonal with
$\mathbb{W}_{ii}=W_i$, $R_i = (X_{i1},\ldots, X_{id})^T$
and
$S_i = R_i Y_i$. 
Also,
$$
H_{ij}\equiv \frac{\partial^2 \hat\beta_1}{\partial v_i \partial v_j} =
-e^T (X^T \mathbb{W} X)^{-1} \Biggl\{ (R_i R_i^T) \frac{\partial \hat\beta}{\partial v_i} +
(R_j R_j^T) \frac{\partial \hat\beta}{\partial v_j}\Biggr\}.
$$
\end{lemma}

The proof is straightforward and is omitted.

\begin{center}
Coordinate Ascent
\end{center}

\begin{enumerate}
\item Input: Data $(B,A,Y)$, where
$B$ is the $n\times k$ matrix with elements
$B_{ij} = b_j(A_i)$,
weights $W_i = 1/\pi(A_i |  X_i)$ and
grid $\{\gamma_1,\ldots,\gamma_N\}$
with $\gamma_1 =1$.
\item Let $\hat\beta = (B^T \mathbb{W} B)^{-1} B^T \mathbb{W} Y$ where
$\mathbb{W}$ is diagonal with $\mathbb{W}_{ii}=W_i$.
Set $\overline{\beta}_1(1)=\underline{\beta}_1(1)=\hat\beta_1$.

\item 
Now move $v_i=1$ to
$v_i = \gamma_2$ or
$v_i = 1/\gamma_2$,
whichever makes $\hat\beta_1$ larger:
\begin{enumerate}
\item Let $G = (B^T \mathbb{W} B)$.
For each $i$ let
\begin{align*}
u_i &= e^T 
\Biggl(G^{-1}_i - \frac{\Delta_i G^{-1}_i r_i r_i^T G^{-1}}
{1 + \Delta_i r_i^T A^{-1}_i r_i}\Biggr)
(B^T \mathbb{W} + \Delta_i e_i e_i^T) Y\ \ \ {\rm flip\ 1\ to\ }\gamma_2\\
\ell_i &= e^T 
\Biggl(G^{-1}_i - \frac{\delta_i G^{-1} r_i r_i^T G^{-1}}
{1 + \delta_i r_i^T G^{-1} r_i}\Biggr)
(B^T \mathbb{W} + \delta_i e_i e_i^T) Y \ \ \ {\rm flip\ 1\ to\ }1/\gamma_2\\
\end{align*}
where
$\Delta_i = \gamma_2 - 1$ and
$\delta_i = \frac{1}{\gamma_2} - 1$.
\item If $u_i \geq \ell_i$:
set $v_i = \gamma_2$ and $I_i = 1$.
Else,
set $v_i = 1/\gamma_2$ and $I_i = 0$.
\end{enumerate}

\item 
For $j=3,\ldots, N$:
Try flipping each $v_i$ to $1/alpha_i$.
\begin{enumerate}
\item Let $v_i = \gamma_j I_i + \gamma_j^{-1} (1-I_i)$.
\item Let $\mathbb{W}_{ii} = v_i$ and
$\hat\beta = (B^T \mathbb{W} B)^{-1} B^T \mathbb{W} Y$.
\item 
Let $A = (B^T \mathbb{W} B)$,
\begin{align*}
t_i &= e^T 
\Biggl(A^{-1}_i - \frac{\Delta_i A^{-1}_i r_i r_i^T A^{-1}_i}
{1 + \Delta_i r_i^T A^{-1}_i r_i}\Biggr)
(B^T \mathbb{W} + \Delta_i e_i e_i^T) Y.
\end{align*}
where
$\Delta_i = 1/v_i - v_i$.
If $t_i > \hat\beta_1$ let $v_i = 1/v_i$.
Let $\mathbb{W}_{ii} = v_i$.
Let $\hat\beta = (B^T \mathbb{W} B)^{-1} B^T \mathbb{W} Y$.
Let $\overline{\beta}_1(\gamma_j) = \hat\beta_1$.
\end{enumerate}
\end{enumerate}

\section{Technical proofs}

\subsection{Proof of Proposition \ref{prop::msm}}
Let $\alpha(u, x, a) = \pi(a|x)/\pi(a|x,u)$. Recall that $v(X, A, Y) = \E\{\alpha(U, X, A) \mid X, A, Y\}$ and $w(A, X) = \pi(A) / \pi(A \mid X)$. We have
\begin{align*}
	0 &= \E\left[ h(A)w(A, X)\{Y-g(A;\beta)\}\alpha(U,X,A)\right]\\
	&= \E\left[ h(A)w(A, X)\{Y-g(A;\beta)\}v(X, A, Y)\right]\\
	&= \E\left\{ h(A)Y v(X,A,Y) - h(A) w(A, X)g(A;\beta)v(X,A,Y)\right\}\\
	&= \E\left[h(A)w(A, X)\E\{Y v(X,A,Y) |  X,A\} - h(A) w(A, X) g(A;\beta)\E\{v(X,A,Y) |  X,A\}\right]\\
	&= \E\left\{ h(A)w(A, X)m(X,A) - h(A) w(A, X)g(A;\beta)\right\}\\
	&= \int \int \{h(a) m(x,a) d\Pb(x) - h(a) g(a;\beta)\} d\Pb(x) \pi(a) da \\
	&= \E\left[ h(A) \left\{\int m(A, x) d\Pb(x) - g(A; \beta)\right\}\right].\ \ \Box
\end{align*}

\subsection{Proof of Lemma \ref{lemma:np_cond_bounds} \label{appendix:proof_lemma_np_cond_bounds}}

We will prove the result for the upper bound. The proof for the
lower bound follows analogously.  We have $v_u(Z) \in
[\gamma^{-1}, \gamma]$ and we can check that $\E\{v_u(Z) |
A, X\} = 1$. Indeed
\begin{align*}
	\E\{v_u(Z)  |  A, X\} & = 
	\gamma \Pb\left( Y > q_u(Y  |  A, X)  |  A, X\right)  + \frac{1}{\gamma} \Pb\left( Y \leq q_u(Y  |  A, X)  |  A, X\right) \\
	& = \gamma \left( 1 - \frac{\gamma}{1 + \gamma} \right) + \frac{1}{\gamma} \cdot \frac{\gamma}{1 + \gamma} = 1,
\end{align*}
because $q_u(A, X)$ is the $\gamma / (1 + \gamma)$-quantile of the conditional distribution of $Y$ given $(A, X)$. Let $v(Z)$ be any
function contained in $[\gamma^{-1}, \gamma]$ such that $\E\{v(Z) |
A, X\} = 1$. We have
\begin{align*}
	\begin{cases}
		v_u(Z) - v(Z) \geq 0 & \text{ if } Y > q_u(Y  |  A, X) \\
		v_u(Z) - v(Z) \leq 0 & \text{ if } Y  \leq q_u(Y  |  A, X)
	\end{cases}
\end{align*}
Therefore, $Y \{v_u(Z) - v(Z)\} \geq q_u(Y  |  A, X)\{v_u(Z) - v(Z)\}$ so that
\begin{align*}
	\E\left\{Y \{v_u(Z) - v(Z)\}  |  A, X \right\} \geq q_u(Y  |  A, X) \E\{v_u(Z)- v(Z)  |  A, X\}= 0
\end{align*}
as desired. 
\subsection{Proof of Proposition \ref{prop:vdv} \label{appendix:proof_prop_vdv}}
This proposition follows directly from Lemma \ref{lemma:z_estimation}, except that we need to show the validity of condition 4. This condition holds under the assumption of Proposition \ref{prop:vdv} because
\begin{align*}
	\mathbb{U}\left[h_l(A_1)\{\widehat{\varphi}_j(Z_1, Z_2) - \varphi_j(Z_1, Z_2)\}\right] & = \int h_l(a) \widehat{w}(a, x)\{\kappa(a, x; \widehat{q}_j) - \kappa(a, x; q_j)\} d\Pb(a, x) \\
	& \hphantom{=} + \int \{w(a, x) - \widehat{w}(a, x)\}\{ \widehat\kappa(a, x; \widehat{q}_j) - \kappa(a, x; q_j) \} d\Pb(a, x)
\end{align*}
Therefore, by Cauchy-Schwarz and Lemma \ref{lemma:sjfacts}:
\begin{align*}
	\left| 	\mathbb{U}\left[h_l(A_1)\{\widehat{\varphi}_j(Z_1, Z_2) - \varphi_j(Z_1, Z_2)\}\right] \right| \lesssim \|q_j - \widehat{q}_j\|^2 + \|w - \widehat{w}\| \| \widehat\kappa_j - \kappa_j\|.
\end{align*}
\subsection{Proof of Proposition \ref{prop::glinear} \label{appendix::proof_glinear}}
We apply Lemma \ref{lemma:z_estimation} to the moment condition
\begin{align*}
\Psi_n(\beta) = \mathbb{U}_n \left[ b(A) \left\{ \widehat{f}_j^s(Z_1, Z_2) - b(A_1)^T\beta \right\} \right] = o_\Pb(n^{-1/2})
\end{align*}
The function class $\mathcal{G}_l = \left\{ a \mapsto b_l(a) b(a)^T \beta, \beta \in \R^k \right\}$ is Donsker since its a finite dimensional vector space (Lemma 7.15 in \cite{sen2018gentle}). Thus, it remains to check condition 4. We have
\begin{align*}
	\left| \mathbb{U} \left\{ \widehat{f}_j^s(Z_1, Z_2) - f_j^s(Z_1, Z_2)\right\} \right| & \lesssim \|q_j - \widehat{q}_j\|^2 + \|w - \widehat{w}\|\|\kappa_j - \widehat\kappa_j\| + \sup_a\left| b_0^T(\widehat{Q}- Q)h(a)\right|^2 \\
	& = o_\Pb(n^{-1/2})
\end{align*}
by assumption and because the last term is $O_\Pb(n^{-1}) = o_\Pb(n^{-1/2})$ by Lemma \ref{lemma:rudelson}.

\subsection{Proof of Lemma~\ref{lemma:F1F2} \label{appendix:proof_lemma_F1F2}}

For
$v\in {\cal V}_{\rm small}(\gamma)$
we have
$\int\int\int h(a) w(a,x)g(a;b)v(z) dP(z) =
\int\int h(a) w(a,x)g(a;b) [\int v(z) p(y|x,a)] \pi(a|x) dP(x)=
\int\int h(a) w(a,x)g(a;b) \pi(a|x) dP(x)=
\int\int h(a) w(a,x)g(a;b) dP(z)$
since
$\E[v(Z)|X,A]=1$.
Therefore $F_1=F_2$
and the result follows.

\subsection{Proof of Lemma \ref{lemma:beta_deriv} \label{appendix:proof_lemma_beta_deriv}}

Let $Z = (A, X, Y)$ and $p(z)$ denote its density. From the moment condition, we have that $F_2(v)$ satisfies
		\begin{align*}
			\int h(a) w(a, x)v(z) y d\Pb(z) = \int h(a)g(a; F_2(v))w(a, x) d\Pb(z)
		\end{align*}
Let $m: \mathcal{F} \mapsto \R$ be a generic functional taking as
input a function $f$. The functional derivative of $m$ with respect to
$f(z)$, denoted $\frac{\delta}{\delta f} m$, satisfies
\begin{align*}
\frac{d}{d\epsilon} m(f + \epsilon \eta) \Large|_{\epsilon = 0} = \int \frac{\delta}{\delta f}m(z) \eta(z) d\Pb(z)
\end{align*}
for any function $\eta$. 
Letting
		\begin{align*}
			\nabla_\beta g(A: \beta) = \begin{bmatrix}
				\frac{d}{d\beta_1} g(A; \beta) \\ \vdots \\ \frac{d}{d \beta_k}g(A; \beta)
			\end{bmatrix}, \quad \frac{d}{d\epsilon} F_2(v + \epsilon \eta) = \begin{bmatrix}
				\frac{d}{d\epsilon} F_{2, 1}(v + \epsilon \eta) \\ \vdots \\ 
				\frac{d}{d\epsilon} F_{2, k}(v + \epsilon \eta)\end{bmatrix}, \quad \text{ and } \quad \frac{\delta}{\delta v} F_2(v) = \begin{bmatrix}
				\frac{\delta}{\delta v} F_{2, 1}(v) \\ \vdots \\ \frac{\delta}{\delta v} F_{2, k}(v)\end{bmatrix}
		\end{align*}
and taking the functional derivative with respect to $v(z)$ on both sides of the expression above yields 
		\begin{align*}
			&  \frac{d}{d\epsilon} 	
			\int h(a)w(a, x) \{v(z) + \epsilon \eta(z)\}yp(z) dz \Large|_{\epsilon = 0} = 
			\int h(a) w(a, x)y\eta(z)p(z) dz \\
			& \implies \frac{\delta}{\delta v }\int h(a) w(a, x) v(z) y p(z) dz = h(a) w(a, x) y \\
			& \frac{d}{d\epsilon} 	\int h(a)w(a, x)g(a; \beta(v + \epsilon\eta)) p(z) dz  \Large|_{\epsilon = 0} = \int h(a)w(a, x)\nabla_\beta g(a; \beta) ^Tp(z) dz \frac{d}{d\epsilon} \beta(v + \epsilon\eta)\Large|_{\epsilon = 0}\\
			& \implies \frac{\delta}{\delta v } \int h(a)w(a, x)g(a; \beta(v)) p(z) dz = \E\left\{ h(A)w(A, X) \nabla_\beta g(A; \beta) ^T\right\} \frac{\delta \beta(v)}{\delta v} 
		\end{align*}
Thus, we conclude that the functional derivative of $\beta(v)$ with respect to $v$ satisfies
		\begin{align*}
			\frac{\delta F_2(v)}{\delta v} = 
			\E\left\{ h(A)w(A, X)\nabla_\beta g(A; \beta)^T \right\} ^{-1} h(a) w(a, x) y
		\end{align*}
as desired. A similar calculation yields $\frac{\delta F_1(v)}{\delta v}$.

\subsection{Proof of Lemma \ref{lemma::FP} \label{appendix:proof_lemma_FP}}

Property 1: This is clear.\\
Property 2: Define a map
$F:{\cal V}(\gamma) \rightarrow {\cal V}(\gamma)$ 
by
$$
F(v) = \gamma I(d_v > q_v) + \frac{1}{\gamma} I(d_v<q_v)
$$
where
$d_v(z) = \delta \beta/\delta v (z)$
and $q_v$ is the $\gamma/(1+\gamma)$ quantile of $d_v(Z)$.
We want to show that there is a fixed point
$v = L(v)$.
Define the metric $m$ by
$m(v_1,v_2) = \sqrt{\int (v_1(z) - v_2(z))^2 d\Pb(z)}$.
The set of functions ${\cal V}(\gamma)$
is a nonempty, closed, convex set.
It is easy to see that
$L:{\cal V}(\gamma) \rightarrow {\cal V}(\gamma)$ is
continuous, that is,
$m(v_n,v)\to 0$ implies
$L(v_n)\to L(v)$.
According to Schauder's fixed point theorem
there exists a fixed point $v_\gamma$
so that $L(v_\gamma)=v_\gamma$.\\
Property 3:
Let $v\in {\cal V}(\gamma)\bigcap B(v_\gamma,\epsilon)$.
Then
$\beta(v) = \beta(v_\gamma) + \int (v(z)-v_\gamma(z)) d_\gamma(z) d\Pb(z) + O(\epsilon^2)$.
The linear functional 
$\int (v(z)-v_\gamma(z)) d_\gamma(z) d\Pb(z)$
is maximized over ${\cal V}(\gamma)$ by choosing
$v = \gamma I(d_\gamma(z) >t) + \gamma^{-1} I(d_\gamma(z) < t)$.
The condition $\int v(z) d\Pb(z)=1$
implies that $t=q$.
So
$\int (v(z)-v_\gamma(z)) d_\gamma(z) d\Pb(z)$
is maximized by $v=v_\gamma$ and hence
$\int (v(z)-v_\gamma(z)) d_\gamma(z) d\Pb(z)\leq 0$.
Thus
$\beta(v_\gamma) \geq \beta(v) + O(\epsilon^2)$. $\Box$

\subsection{Proof of Lemma \ref{lemma:linB}.}

The fact that $F_1$ and $F_2$ yield the same bounds
follows from Lemma 6.
Now
$F_2(v) = \int v(z) q(z) dP(z)$
where $q(z) = y w(a,x)M^{-1} b(a)$
which is a linear functional.
The form of the maximizer amd minimizer follows by the same argument as in the proof of Lemma 2.

\subsection{Proof of Lemma \ref{lemma::largeiseasy}.}

Since $F_2(v)$
is a linear functional of $v$,
The form of the maximizer and minimizer follows by the same argument as in the proof of Lemma 2.

\subsection{Proof of Lemma \ref{prop:outcome_model}}\label{app:proof_outcome_model}

Consider the upper bound. We apply Lemma \ref{lemma:z_estimation} to the moment condition
\begin{align*}
\Psi_n(\widehat\beta) = \mathbb{U}_n \left[ b(A_1) \left\{\widehat{\zeta}_u(Z_1, Z_2) - b(A_1)^T\widehat\beta \right\}\right] = o_\Pb(n^{-1/2})
\end{align*}
where 
$\zeta_u(Z_1, Z_2) = w(A_1, X_1)\left\{ Y_1 - \mu(A_1, X_1)  \right\} + \mu(A_1, X_2) + \delta \sgn\left\{ b(a_0)^TQ^{-1} b(A_1) \right\}$ and $a_0$ is a fixed value of $a$ that we want to distinguish from
the dummy $a$ in  the function class $\mathcal{G}_l = \left\{ a \mapsto b_l(a)b(a)^T\beta, \beta \in \R^k \right\}$.  Notice that $\mathcal{G}_l$ is Donsker and we have
$\dot{\Psi}_{\beta_0} = Q = \E\{b(A)b^T(A)\}$.
Next notice that, by virtue of the statement of Lemma \ref{lemma:z_estimation}:
\begin{align*}
\widehat{g}_u(a_0) - g_u(a_0) = 
b(a_0)^T (\widehat\beta - \beta_u) & = b(a_0)^T Q^{-1} \mathbb{U}b(A_1) \left\{\widehat{\zeta}_u(Z_1, Z_2) - \zeta_u(Z_1, Z_2)\right\} \\
& \hphantom{=} + b(a_0)^T Q^{-1} (\mathbb{U}_n - \mathbb{U}) b(A_1)\left\{\zeta_u(Z_1, Z_2) - b^T(A_1) \beta_u\right\} + o_\Pb(n^{-1/2}).
\end{align*}
Next, we have
\begin{align*}
& \left| \mathbb{U} \left\{b(a_0)^TQ^{-1} b(A_1) \widehat{\zeta}_u(Z_1, Z_2) - \zeta_u(Z_1, Z_2) \right\} \right|\\
&  \lesssim \sup_a| b(a_0)^TQ^{-1} b(a)| \| w - \widehat{w}\| \|\mu - \widehat\mu\| \\
& \hphantom{\lesssim} + \left| \mathbb{P} \left(b(a_0)^TQ^{-1} b(A) \left[\sgn\left\{b(a_0)^T\widehat{Q}^{-1}b(A)\right\} - 
\sgn\left\{b(a_0)^TQ^{-1}b(A)\right\} \right]\right) \right|.
\end{align*}
By assumption the first term is $o_\Pb(n^{-1/2})$. By Lemma \ref{lemma:margin}, the last term is upper bounded by a constant multiple of
\begin{align*}
\sup_{a} \left| b(a_0)^T(\widehat{Q}^{-1} - Q^{-1}) b(a) \right|^2
\end{align*}
which is $O_\Pb(n^{-1})$ by Lemma \ref{lemma:rudelson}.

\subsection{Influence Function for $\beta(v_\gamma)$ \label{app::influ}}

The parameter is
$\psi = \beta(v)$
where
$v$ is given by the fixed point equation
$$
v(z) = \gamma - \left(\gamma - \gamma^{-1}\right) I ( d(z) - q<0).
$$
Now $v$ is a function of $p$ and $z$
and $d$ is a function of $v$ and $z$
so we will write
$v=v(p,z)$ and $d = d(v(p),z)$
and
$$
v(p,s) = \gamma - \left(\gamma - \gamma^{-1}\right) I ( d(v(p),z) - q < 0).
$$
The influence function is not well-defined beacause of the presence of the indicator 
function.
So we approximate $v$
by
$$
v(p,z) = \gamma - \left(\gamma - \gamma^{-1}\right) S ( d(v(p),z) - q)
$$
where $S$ is any smooth approximation to the indicator function.
In general,
the influence function $\varphi(z)$ of a parameter $\psi$
is relate to the $L_2(P)$ functional derivative by
$\varphi(z) = (1/p(z)) \delta \psi (z)/\delta p$.
We then have
$$
\frac{\delta \beta(v(p))}{\delta p} =
\int \frac{\delta \beta(v(p))}{\delta v(p,s)} \frac{\delta v(p,z)}{\delta p} d\Pb(z) = 
\int d_\gamma(z)  \frac{\delta v(p,z)}{\delta p} d\Pb(z).
$$
Now
\begin{align*}
\frac{\delta v(p,z)}{\delta p}(Z) &=
- \left(\gamma - \gamma^{-1}\right) S'(d(v(p),z) - q) 
\left(\frac{\delta d(v(p),z)}{\delta p}(Z) - \frac{\delta q}{\delta p}(Z)\right)\\
&=
- \left(\gamma - \gamma^{-1}\right) S'(d(v(p),z) - q) 
\left(\int \frac{\delta d(v(p),z)}{\delta v(p,t)}(Z) \frac{\delta v(p,t)}{\delta p}(Z) d\Pb(t) - \frac{\delta q}{\delta p}(Z)\right)\\
\end{align*}
Note that
$\delta v/\delta p$ appears on both sides
and so the influence function involves solving an integral equation.

We still need to find
$\delta d(v(p),z)/\delta v(p,t)(Z)$ 
and
$\frac{\delta q}{\delta p}(Z)$.
We may write the formula for $d(z)$ as
$$
d_\gamma(z)\int h(a) g'(a,\beta)  r(x,a,y) = h(a)y w(a,x)v(z)
$$
where
$r = p(x)\pi(a)p(y|x,a)$ and $W=\pi(a)/\pi(a|x)$.
Note that
$$
\mathring{r} = 
\delta_a p(x)p(y|x,a) + \delta_x \pi(a)p(y|x,a) + 
\pi(a)\delta_{xa}\frac{ \delta_y - p(y|x,a)}{\pi(a|x)} - p(x)\pi(a)p(y|x,a)
$$
and
$$
\mathring{W} =
\frac{\delta_a p(x) + \delta_x \pi(a) - W\delta_{xa}}{p(x)\pi(a|x)} - W
$$
where $\mathring{r}$ means the influence function of $r$ etc.
So
$$
\mathring{d}\int h(a)g' r + 
d(z)\int h(a) \mathring{g'} r + 
d(z) \int h(a) g' \mathring{r} = h(a) y \mathring{W} v(z)
$$
and therefore
$$
\mathring{d} = 
(\int h(a)g' r)^{-1}
h(a) y \mathring{W} v(z) - d(z)\int h(a) \mathring{g'} r - d(z) \int h(a) g' \mathring{r}
\ \ \text{ and } \ \ 
\frac{\delta d(v(p),z)}{\delta v(p,t)}(Z) =
\frac{\mathring{d}(Z)}{p(Z)}.
$$
To find $\mathring{q}$
note that
$F(q,p) = \gamma/(1+\gamma)$
where
$F(t,p) = P(d(Z) \leq t)$.
So
$
f(q)\mathring{q} + \mathring{F} = 0, 
$
which implies
$\mathring{q} = - \mathring{F}/f(q)$.
Now
$$
F(t,p) = \int I(d(z,p) \leq t) p(z) dz \ \ \text{    and    } \ \  \mathring{F}(t,p) =
I(d_\gamma(z)\leq t)- \int I(d_\gamma(z) = t)\mathring{d}_\gamma(z)p(z) dz,
$$
so that
$
\mathring{F}(q,p) =
I(d_\gamma(z)\leq q)- \int I(d_\gamma(z) = q)\mathring{d}_\gamma(z)p(z) dz.
$
Hence
$$
\mathring{q} = - 
\frac{I(d_\gamma(z)\leq q)- \int I(d_\gamma(z) = q)\mathring{d}_\gamma(z)p(z) dz}{f(q)}
\ \ \text{ and } \ \ 
\frac{\delta q}{\delta p}(Z) = \frac{\mathring{q}(Z)}{p(Z)}.
$$
Finally,
$$
\mathring{g}'(a,\beta) =
p(z)\frac{\delta g'(a,\beta)}{\delta p} =
p(z) \int \frac{\delta g'(a,\beta)}{\delta v_\gamma} 
\frac{\delta v_\gamma}{\delta p}=
p(z) \int d_\gamma(z) \mathring{v}(z) dz.
$$

\subsection{Proof of Proposition \ref{prop::bound_g_epsilon}}\label{appendix::proof_bound_g_epsilon}
We will prove the proposition in two steps:
\begin{enumerate}
	\item We show that Lemma \ref{lemma:z_estimation} yields that 
	\begin{align*}
		\widetilde\beta_j - \beta_j = - \dot\Psi^{-1}_{\beta_j}(\mathbb{U}_n - \mathbb{U}) h(A_1)\{f(Z_1, Z_2) - g(A_1; \beta_j)\} + o_\Pb(n^{-1/2}) 
	\end{align*}
where $\widetilde\beta_j$ solves:
\begin{align*}
	& \mathbb{U}_n h(A_1)\left[ \widetilde{f}(Z_1, Z_2) - g(A_1; \widetilde\beta)\right] = o_\Pb(n^{-1/2}), \text{ where } \\
	& \widetilde{f}(Z_1, Z_2) = \widehat{f}_\mu(Z_1, Z_2) + \lambda_j(A_1, X_1)\widehat{f}_\Delta(Z_1) + \lambda(A_1, X_2)\widehat{f}_r(Z_1, Z_2).
\end{align*}
That is, $\widetilde\beta$ solves the original moment condition except that the estimator of the indicator term is replaced with the true indicator , e.g. $\widehat\lambda_u(a, x) = \one\{\widehat{r}_u(a, x) > \widehat{t}_{\epsilon, u}\}$ is replaced by $\lambda_u(a, x) = \one\{r_u(a, x) > t_{\epsilon, u}\}$. 
	\item We show that $$\widehat\beta_j - \widetilde\beta_j = -\dot\Psi^{-1}_{\widehat\beta_j}(\mathbb{U}_n - \mathbb{U}) h(A_1) t_{\epsilon, j}(A_1) \lambda_j(A_1, X_2) + o_\Pb(n^{-1/2}).$$
\end{enumerate}
From these statements, it follows by Lemma \ref{lemma:u_vdv} and Slutsky's theorem, that
\begin{align*}
	\sqrt{n}(\widehat\beta_j - \beta_j) \indist N(0, 4\Sigma)
\end{align*}
because, by the continuous mapping theorem, $\dot\Psi^{-1}_{\widehat\beta} \inprob \dot\Psi^{-1}_{\beta_j}$ since $\widehat\beta \inprob \beta_j$. 
\subsubsection{Step 1}
Because $\widetilde{f}(Z_1, Z_2)$ is fixed given the training sample, we can apply Lemma \ref{lemma:z_estimation}. In particular, all the conditions of the lemma are satisfied by assumption and by noticing that
\begin{align*}
	\left|\mathbb{U} \left\{\widehat{f}_\mu(Z_1, Z_2) - f_\mu(Z_1, Z_2)\right\}\right| \lesssim \| w - \widehat{w}\|\|\mu- \widehat\mu\|
\end{align*}
and
\begin{align*}
& \left| \mathbb{U} \left[  \lambda_j(A_1, X_1)\widehat{f}_\Delta(Z_1) + 
\lambda(A_1, X_2)\left\{\widehat{f}_r(Z_1, Z_2) - f_r(Z_1, Z_2)\right\}\right] \right| \\
& \lesssim \| w - \widehat{w}\|\left( \| \kappa_j - \widehat\kappa_j\| + \|\mu- \widehat\mu\| \right) + \|q_j - \widehat{q}_j\|^2.
\end{align*}
Therefore, condition 4 in Lemma \ref{lemma:z_estimation} is satisfied
as well under the assumption that the nuisance functions are estimated
with enough accuracy.

\subsubsection{Step 2}
Define $\widetilde\lambda_\ell(a, x) = \one\left\{ \widehat{r}_\ell(a, x) \leq t_{\epsilon, l}(a) \right\}$, $\widetilde\lambda_u(a, x) = \one\left\{ \widehat{r}_u(a, x) > t_{\epsilon, u}(a) \right\}$. First notice that, by construction of $\widehat{t}_{\epsilon, j}(a)$, for every $a \in \mathcal{A}$:
\begin{align*}
	o_\Pb(n^{-1/2}) = \Pn \widehat\lambda_j(a, X) - \Pb\lambda_j(a, X)
\end{align*}
where $\Pn \widehat\lambda_j(a, X)$ is the sample average over the test sample used to construct the $U$-statistics. In this light, $\mathbb{U}_nh_l(A_1)t_{\epsilon, j}(A_1)\widehat\lambda_j(A_1, X_2) - \mathbb{U}h_l(A_1)t_{\epsilon, j}(A_1)\lambda_j(A_1, X_2) = o_\Pb(n^{-1/2})$ and 
\begin{align*}
	o_\Pb(n^{-1/2}) & = ( \mathbb{U}_n - \mathbb{U})\left[h(A_1)t_{\epsilon, j}(A_1)\left\{\widehat\lambda_j(A_1, X_2) - \widetilde\lambda_j(A_1, X_2) \right\} \right]\\
	& \hphantom{=} + ( \mathbb{U}_n - \mathbb{U})\left[h(A_1)t_{\epsilon, j}(A_1)\left\{\widetilde\lambda_j(A_1, X_2) - \lambda_j(A_1, X_2) \right\} \right] \\
	& \hphantom{=} + ( \mathbb{U}_n - \mathbb{U})\left[h(A_1)t_{\epsilon, j}(A_1)\lambda_j(A_1, X_2)\right] + \mathbb{U}\left[t_{\epsilon, j}(A_1)h_l(A_1)\left\{\widehat\lambda_j(A_1, X_2) - \lambda_j(A_1, X_2) \right\}\right] .
\end{align*}
Notice that the middle term involving $\widetilde\lambda_j(A_1, X_2) -
\lambda_j(A_1, X_2)$ is an empirical process term of a fixed function
given the training sample. Therefore, by Lemma \ref{lemma:edward_u},
it is $o_\Pb(n^{-1/2})$ because
\begin{align*}
	& \int \left| S_2\left\{\widetilde\lambda_j(a_1, x_2) - \lambda_j(a_1, x_2) \right\}\right|d\Pb(z_2) & \\
	& \leq \int \one\left\{ |r_j(a_1, x_2) -t_{\epsilon, j}(a_1) | \leq \| \widehat{r}_j - r_j\|_\infty \right\} d\Pb(x_2) + \int \one\left\{ |r_j(a_2, x_1) - t_{\epsilon, j}(a_2) | \leq \| \widehat{r}_j - r_j\|_\infty \right\} d\Pb(a_2) \\
	& \lesssim \| \widehat{r}_j - r_j\|_\infty \\
	& = o_\Pb(1)
\end{align*}
because the densities of $r_j(a, X) - t_{\epsilon, j}(a)$ and $r_j(A, x) - t_{\epsilon, j}(A)$ are assumed to be bounded for any $a$ and $x$. In this respect, we have
\begin{align*}
	o_\Pb(n^{-1/2}) & = ( \mathbb{U}_n - \mathbb{U})\left[h(A_1)t_{\epsilon, j}(A_1)\left\{\widehat\lambda_j(A_1, X_2) - \widetilde\lambda_j(A_1, X_2) \right\} \right]\\
	& \hphantom{=} + ( \mathbb{U}_n - \mathbb{U})\left[h(A_1)t_{\epsilon, j}(A_1)\lambda_j(A_1, X_2)\right] + \mathbb{U}\left[t_{\epsilon, j}(A_1)h_l(A_1)\left\{\widehat\lambda_j(A_1, X_2) - \lambda_j(A_1, X_2) \right\}\right] 
\end{align*}

Because both $\widehat{\beta}$ and $\widetilde\beta$ solve empirical moment conditions, we have
\begin{align*}
	o_\Pb(n^{-1/2}) & = \mathbb{U}_n\left\{ \widehat{f}(Z_1, Z_2) - \widetilde{f}(Z_1, Z_2) \right\} +  \Pn h(A)\left\{ g(A; \widetilde\beta) - g(A; \widehat\beta)\right\}
\end{align*}
and, in light of the observations above, we can subtract the $o_\Pb(n^{-1/2})$ term to obtain
\begin{align*}
		o_\Pb(n^{-1/2})  & = (\Pn - \Pb) h(A_1) \widehat{f}_\Delta(Z_1) \left\{ \widehat\lambda_j(A_1, X_1) - \widetilde\lambda_j(A_1, X_1)\right\} \\
		& \hphantom{=}  + (\mathbb{U}_n - \mathbb{U}) \left[h(A_1) \left\{ \widehat\lambda_j(A_1, X_2) - \widetilde\lambda_j(A_1, X_2)\right\} \left\{\widehat{f}_r(Z_1, Z_2) - t_{\epsilon, j}(A_1) \right\} \right] \\
	& \hphantom{=} - ( \mathbb{U}_n - \mathbb{U})\left[t_{\epsilon, j}(A_1)h(A_1)\lambda_j(A_1, X_2)\right] \\
	& \hphantom{=} +  \Pb \left[h(A_1) \left\{ \widehat\lambda_j(A_1, X_1) - \lambda_j(A_1, X_1)\right\} \widehat{f}_\Delta(Z_1) \right] \\
	&  \hphantom{=} + \mathbb{U}\left[h(A_1) \left\{ \widehat\lambda_j(A_1, X_2) - \lambda_j(A_1, X_2)\right\} \left\{\widehat{f}_r(Z_1, Z_2) - t_{\epsilon, j}(A_1) \right\} \right] \\
	& \hphantom{=} + (\Pn - \Pb)h(A)\left\{ g(A; \widetilde\beta) - g(A; \widehat\beta)\right\} + \dot\Psi_{\widehat\beta} (\widetilde\beta - \widehat\beta) + o(\|\widetilde\beta - \widehat\beta\|)
\end{align*}
where we used the identity
\begin{align*}
	\Pn h(A)\left\{ g(A; \widetilde\beta) - g(A; \widehat\beta)\right\} & =  (\Pn - \Pb)h(A)\left\{ g(A; \widetilde\beta) - g(A; \widehat\beta)\right\} + \dot\Psi_{\widehat\beta} (\widetilde\beta - \widehat\beta) + o(\|\widetilde\beta - \widehat\beta\|).
\end{align*}
Next, we claim that, conditioning on the training sample $D^n$ and thus viewing $\widehat{f}_\Delta(z)$ and $\widehat{r}_j(a, x)$ as fixed functions,  the function class $\mathcal{F}= \left\{ f(z) = h_j(a)\widehat{f}_\Delta(z) \one\left\{ \widehat{r}_j(a, x) - t_{\epsilon, j}(a) > 0 \right\}, t_{\epsilon, j}(a) \in \mathcal{T} \right\}$ is VC-subgraph. The subgraph $\mathcal{C}_q$ of $f_t(z) \equiv \one\left\{ \widehat{r}_j(a, x) - t_{\epsilon, j} > 0 \right\}$ is the collection of sets $(z, c)$ in $\mathcal{Z} \times \R$  such that $f_t(z) \geq c$. For a given $t\equiv t_{\epsilon, j}$, let $S_0(t)$ be the collection of all $z$ such that $\widehat{r}_j(a, x) - t_{\epsilon, j}(a) \leq 0$. Then, we have that the subgraph of $f_t(z)$ is
\begin{align*}
	S_0(t) \times (-\infty, 0] \cup S_0^c(t) \times (-\infty, 1]
\end{align*}
By Lemma 7.19 (iii) in \cite{sen2018gentle}, $S_0(t)$ is a VC set whenever $\widehat{r}_j(a, x) - t_{\epsilon, j}(a)$ is VC-subgraph, which is the case since $t_{\epsilon, j}(a)$ is VC-subgraph by assumption and $\widehat{r}_j(a, x)$ is a fixed function (given the training data). This then yields that the subgraph of $f_t(z)$ is a VC-set. Because $\mathcal{F}$ consists of products of VC-subgraph functions and $h_l(a)\widehat{f}_\Delta(z)$, a fixed function, we conclude that $\mathcal{F}$ itself is a VC-subgraph class. This means that the process $\sqrt{n}(\Pn - \Pb) f$, $f \in \mathcal{F}$, is stochastically equicontinuous relative to $\rho(f_1, f_2) = [\var\{f_1(Z) - f_2(Z)\}]^{1/2} \leq \|f_1 - f_2\|$. Thus,
\begin{align*}
 (\Pn - \Pb) h(A_1) \widehat{f}_\Delta(Z_1) \left\{ \widehat\lambda_j(A_1, X_1) - \widetilde\lambda_j(A_1, X_1)\right\} = o_\Pb(n^{-1/2})
\end{align*}
because, using the assumption that $r_j(A, X) - t_{\epsilon, j}(A)$ has a bounded density:
\begin{align*}
	\int \left[h_j(a) \widehat{f}_\Delta(z) \left\{ \widehat\lambda_j(a, x) - \widetilde\lambda_j(a, x)\right\}\right]^2d\Pb(z) & \lesssim \int \left| \widehat\lambda_j(a, x) - \widetilde\lambda_j(a, x)\right| d\Pb(a, x) \\
	& \leq \int \one\left\{ |\widehat{r}_j(a, x) - t_{\epsilon, j}(a)| \leq |\widehat{t}_{\epsilon, j}(a) - q(a)| \right\} d\Pb(a, x) \\
	& \leq \int \one\left\{ |r_j(a, x) - t_{\epsilon, j}(a)| \leq \| \widehat{r}_j - r_j \|_\infty+ \|\widehat{t}_{\epsilon, j} - t_{\epsilon, j}\|_\infty \right\} d\Pb(a, x) \\
	& \lesssim  \| \widehat{r}_j - r_j \|_\infty+ \|\widehat{t}_{\epsilon, j} - t_{\epsilon, j}\|_\infty \\
	&  = o_\Pb(1)
\end{align*}
To analyze the empirical $U$-process, we rely on \cite{arcones1993limit}. In particular, by their Theorem 4.9 applied in conjuction with their Theorem 4.1, the process $\sqrt{n}(\mathbb{U}_n - \mathbb{U}) f$, for
\begin{align*}
f \in \mathcal{F} = \left\{ f(z_1, z_2) \mapsto h_l(a_1)\{\widehat{f}_r(z_1, z_2) - t_{\epsilon, j (a_1)}\} \one\left\{ \widehat{r}_j(a_1, x_2) - \overline{t}_{\epsilon, j}(a_1) > 0 \right\}, \overline{t}_{\epsilon, j}(a_1) \in \mathcal{T} \right\}
\end{align*}
is stochastically equicontinuous, relative to the norm 
$$\rho^2(f_1, f_2) = \int \left[ \int S_2\left\{f_1(z_1, z_2) - f_2(z_1, z_2)\right\} d\Pb(z_2)\right]^2 d\Pb(z_1),$$
if, for instance, the class $\mathcal{F}$ is VC-subgraph. This is indeed the case under the assumption that $\mathcal{T}$ is a VC-subgraph class. Let $\widetilde{t}_{\epsilon, j}(z_1, z_2) \equiv \overline{t}_{\epsilon, j}(a_1)$ and $\mathcal{C}_t$ the subgraph of $a \mapsto \overline{t}_{\epsilon, j}(a)$. Then the subgraph of $\widetilde{t}$ is simply $\mathcal{Z} \cap \mathcal{C}_t \times \mathcal{Z}$, which is still a VC set. Then, as argued earlier, $\mathcal{F}$ consists of functions that are products of VC-subgraph classes and thus it is VC-subgraph. This concludes our proof that
\begin{align*}
	(\mathbb{U}_n - \mathbb{U}) \left[h(A_1) \left\{ \widehat\lambda_j(A_1, X_2) - \widetilde\lambda_j(A_1, X_2)\right\} \left\{\widehat{f}_r(Z_1, Z_2) - t_{\epsilon, j}(A_1) \right\} \right] = o_\Pb(n^{-1/2})
\end{align*}
since
\begin{align*}
	\left|\int S_2\left\{ \widehat\lambda_j(a_1, x_2) - \widetilde\lambda_j(a_1, x_2)\right\} d\Pb(z_2) \right| \lesssim \|\widehat{r}_j - r_j\|_\infty + \|\widehat{t}_{\epsilon, j} - t_{\epsilon, j}\|_\infty = o_\Pb(1).
\end{align*}
Next, we have by Cauchy-Schwarz
\begin{align*}
	& \left| \Pb \left[h_l(A) \left\{ \widehat\lambda_j(A, X) - \lambda_j(A, X)\right\} \widehat{w}(A, X)\left\{ \kappa(A, X; \widehat{q}_j) - \widehat\kappa(A, X; \widehat{q}_j) - \mu(A, X)- \widehat\mu(A, X) \right\}\right] \right| \\
	& \lesssim \int \left| \widehat\lambda_j(a, x) - \lambda_j(a, x) \right| d\Pb(a, x)\left(\| \kappa_j - \widehat\kappa_j\| + \|\widehat{q} - q\|^2 + \|\mu - \widehat\mu\|\right) \\
	& \lesssim \left(\| \widehat{r}_j - r_j\|_\infty + \|\widehat{t}_{\epsilon, j} - t_{\epsilon, j}\|_\infty\right)\left(\| \kappa_j - \widehat\kappa_j\| + \|\widehat{q} - q\|^2 + \|\mu - \widehat\mu\|\right) \\
	& = o_\Pb(n^{-1/2})
\end{align*} 
by assumption. This concludes our proof that $\Pb h(A_1) \left\{ \widehat\lambda_j(A_1, X_1) - \lambda_j(A_1, X_1)\right\} \widehat{f}_\Delta(Z_1) = o_\Pb(n^{-1/2})$. 

Next, we have
\begin{align*}
	& \left| \mathbb{U}\left[h_l(A_1) \left\{ \widehat\lambda_j(A_1, X_2) - \lambda_j(A_1, X_2)\right\} \left\{\widehat{f}_r(Z_1, Z_2) - f_r(Z_1, Z_2) \right\} \right] \right| \\
	& = \left| \int h_l(A_1) \left\{ \widehat\lambda_j(a, x) - \lambda_j(a, x)\right\} \left\{\widehat\kappa(a, x; \widehat{q}_j) - \kappa(a, x; \widehat{q}) - \widehat\mu(a, x) - \mu(a, x)\right\} d\Pb(a)d\Pb(x) \right| \\
	& \lesssim (\|\widehat{r}_j - r_j\|_\infty + \|\widehat{t}_{\epsilon, j} - t_{\epsilon, j} \|_\infty)(\|\widehat\kappa_j - \kappa_j\| + \|\widehat\mu - \mu\| + \|\widehat{q} - q\|^2)
\end{align*}
and 
\begin{align*}
	& \left| \mathbb{U}\left[h_l(A_1) \left\{ \widehat\lambda_j(A_1, X_2) - \lambda_j(A_1, X_2)\right\} \left\{f_r(Z_1, Z_2) - t_{\epsilon, j}(A_1) \right\} \right] \right| \\
	& = \left| \int h_l(a) \{\widehat\lambda_j(a, x) - \lambda_j(a, x)\}\{r_j(a, x) - t_{\epsilon, j}(a)\} d\Pb(a)d\Pb(x) \right| \\
	& \lesssim \int \one\{ |r_j(a, x) - t_{\epsilon, j}(a)| \leq \| \widehat{r}_j - r_j \|_\infty + \|t_{\epsilon, j} - \widehat{t}_{\epsilon, j} \|_\infty \} \{r_j(a, x) - t_{\epsilon, j}(a)\} d\Pb(a)d\Pb(x) \\
	& \leq \left(\| \widehat{r}_j - r_j \|_\infty + \|t_{\epsilon, j} - \widehat{t}_{\epsilon, j} \|_\infty\right) \int\Pb\left(|r(a, X) - t_{\epsilon, j}(a)| \leq \| \widehat{r}_j - r_j \|_\infty + \|t_{\epsilon, j} - \widehat{t}_{\epsilon, j} \|_\infty\right) d\Pb(a) \\
	& \lesssim \| \widehat{r}_j - r_j \|^2_\infty + \|t_{\epsilon, j} - \widehat{t}_{\epsilon, j} \|^2_\infty
\end{align*}
This concludes our proof that
\begin{align*}
	 \mathbb{U}\left[h_l(A_1) \left\{ \widehat\lambda_j(A_1, X_2) - \lambda_j(A_1, X_2)\right\} \left\{\widehat{f}_r(Z_1, Z_2) - t_{\epsilon, j} (A_1) \right\} \right] = o_\Pb(n^{-1/2})
\end{align*}
Statement 2 now follows if we can show that 
\begin{align*}
	(\Pn - \Pb)h(A)\left\{g(A; \widetilde\beta_j) - g(A; \widehat\beta_j)\right\} = o_\Pb(n^{-1/2})
\end{align*}
which is the case if $\|\widehat\beta_j - \widetilde\beta_j\| \leq \|\widehat\beta_j - \beta_j\| + \|\widetilde\beta_j - \beta_j\| = o_\Pb(1)$ because $g(A; \beta)$, $\beta \in \R^k$ is a Donsker class. We can show consistency of $\widehat\beta_j$ for $\beta_j$ by relying on Theorem 2.10 in \cite{kosorok2008introduction} as done in the proof of Statement 1 of Lemma \ref{lemma:z_estimation}. Let 
$\widehat\Psi_n(\beta) = \mathbb{U}_nh(A_1)\{ \widehat{f}_j(Z_1, Z_2) - g(A_1; \beta)\}$ and $\Psi(\beta) = \mathbb{U}h(A_1)\{f_j(Z_1, Z_2) - g(A_1; \beta_j)\}$. First, we need to show that $\|\Psi(\beta_n)\| \to 0$ implies $\|\beta_n - \beta_j\| \to 0$ for any sequence $\beta_n \in \R^k$. This is accomplished as in the proof of Lemma \ref{lemma:z_estimation} by differentiability of $\Psi(\beta): \R^k \to \R^k$ and invertibility of its Jacobian matrix:
\begin{align*}
	\Psi(\beta_n) = \dot\Psi_{\beta_j}(\beta_n - \beta_j) + o(\|\beta_n - \beta_j\|) \implies \|\beta_n - \beta_j\|\{1 + o(1)\} \lesssim \|\Psi(\beta_n)\| \to 0. 
\end{align*}
Second, we need to show that $\sup_{\beta \in \R^k} \|\Psi_n(\beta) - \Psi(\beta)\| = o_\Pb(1)$, which is the case since
\begin{align*}
	\Psi_n(\beta) - \Psi(\beta) & = (\mathbb{U}_n - \mathbb{U}) h(A_1)\{\widehat{f}_j(Z_1, Z_2) - f_j(Z_1, Z_2)\} + (\mathbb{U}_n - \mathbb{U}) f_j(Z_1, Z_2) \\
	& \hphantom{=} + \mathbb{U}h(A_1)\{\widehat{f}_j(Z_1, Z_2) - f_j(Z_1, Z_2)\} + (\Pn - \Pb) h(A)g(A; \beta)
\end{align*}
All the terms above are $o_\Pb(1)$ by the arguments made in proving the previous steps and because $g(a; \beta)$, $\beta \in \R^k$ is Donsker and thus Glivenko-Cantelli. This concludes our proof that
\begin{align*}
	\widehat\beta_j - \widetilde\beta_j = -\dot\Psi^{-1}_{\widehat\beta_j}(\mathbb{U}_n - \mathbb{U}) h(A_1) t_{\epsilon, j}(A_1) \lambda_j(A_1, X_2) + o_\Pb(n^{-1/2})
\end{align*}
\newpage
\subsection{Moment condition in the time-varying case} \label{appendix:mc_time_varying}
We assume that $Y(\overline{a}_T) \ind A_t \mid \overline{A}_{t-1}, \overline{X}_t, \overline{U}_t$. Then, we have, for $p(\cdot)$ denoting generically a density:
\begin{align*}
	& \E\left[ h(\overline{A}_T) W_T(\overline{A}_T, \overline{X}_T)\left\{ Yv_T(Y, \overline{A}_T, \overline{X}_T) - g(\overline{A}_T; \beta) \right\} \right] \\
	& = \int \frac{h(\overline{a}_T) \pi(\overline{a}_T)}{\prod_{s = 1}^T \pi(a_s \mid \overline{x}_s,  \overline{a}_{s-1})} \left\{ y v_T(y, \overline{a}_T, \overline{x}_T) - g(\overline{a}_T; \beta) \right\} p(y, \overline{a}_T, \overline{x}_T) dy d\overline{a}_Td \overline{x}_T \\
	& = \int \frac{h(\overline{a}_T) \pi(\overline{a}_T)}{\prod_{s = 1}^T \pi(a_s \mid \overline{x}_s,  \overline{a}_{s-1})} \left\{ y \int \frac{\prod_{s = 1}^T \pi(a_s \mid \overline{a}_{s-1} \overline{x}_s)}{\prod_{s = 1}^T \pi(a_s \mid \overline{a}_{s-1}, \overline{x}_s, \overline{u}_s)} d\Pb(\overline{u}_T \mid \overline{a}_T, \overline{x}_T, y) - g(\overline{a}_T; \beta) \right\} p(y, \overline{a}_T, \overline{x}_T) dy d\overline{a}_Td \overline{x}_T \\
	& =  \int \frac{h(\overline{a}_T)\pi(\overline{a}_T)}{\prod_{s = 1}^T \pi(a_s \mid \overline{x}_s,  \overline{a}_{s-1})} \left\{  \int y \frac{\prod_{s = 1}^T \pi(a_s \mid \overline{a}_{s-1} \overline{x}_s)}{\prod_{s = 1}^T \pi(a_s \mid \overline{a}_{s-1}, \overline{x}_s, \overline{u}_s)} d\Pb(\overline{u}_T, y \mid \overline{a}_T, \overline{x}_T) - g(\overline{a}_T; \beta) \right\} p( \overline{a}_T, \overline{x}_T) d\overline{a}_Td \overline{x}_T \\
	& =  \int \frac{h(\overline{a}_T) \pi(\overline{a}_T)}{\prod_{s = 1}^T \pi(a_s \mid \overline{x}_s,  \overline{a}_{s-1})} \left\{  \int \E(Y^{\overline{a}_T} \mid \overline{a}_T, \overline{x}_T, \overline{u}_T) \frac{\prod_{s = 1}^T \pi(a_s \mid \overline{a}_{s-1} \overline{x}_s)}{\prod_{s = 1}^T \pi(a_s \mid \overline{a}_{s-1}, \overline{x}_s, \overline{u}_s)} d\Pb(\overline{u}_T \mid \overline{a}_T, \overline{x}_T) - g(\overline{a}_T; \beta) \right\} 	\\
	& \hphantom{= \int } p( \overline{a}_T, \overline{x}_T) d\overline{a}_Td \overline{x}_T
\end{align*}
Next, because $Y^{\overline{a}_T} \ind A_T \mid \overline{X}_T, \overline{U}_T, \overline{A}_{T-1}$ and by Bayes' rule, we can further simplify:
\begin{align*}
	& =  \int \frac{h(\overline{a}_T) \pi(\overline{a}_T)}{\prod_{s = 1}^T \pi(a_s \mid \overline{x}_s,  \overline{a}_{s-1})} \left\{  \int \E(Y^{\overline{a}_T} \mid \overline{a}_{T-1}, \overline{x}_T, \overline{u}_T) \frac{\prod_{s = 1}^{T-1} \pi(a_s \mid \overline{a}_{s-1} \overline{x}_s)}{\prod_{s = 1}^{T-1} \pi(a_s \mid \overline{a}_{s-1}, \overline{x}_s, \overline{u}_s)} d\Pb(\overline{u}_T \mid \overline{a}_{T-1}, \overline{x}_T) - g(\overline{a}_T; \beta) \right\} \\
	& \hphantom{= \int }  p( \overline{a}_T, \overline{x}_T) d\overline{a}_Td \overline{x}_T \\
	& =  \int \frac{h(\overline{a}_T) \pi(\overline{a}_T)}{\prod_{s = 1}^T \pi(a_s \mid \overline{x}_s,  \overline{a}_{s-1})} \left\{  \int \E(Y^{\overline{a}_T} \mid \overline{a}_{T-1}, \overline{x}_T, \overline{u}_{T-1}) \frac{\prod_{s = 1}^{T-1} \pi(a_s \mid \overline{a}_{s-1} \overline{x}_s)}{\prod_{s = 1}^{T-1} \pi(a_s \mid \overline{a}_{s-1}, \overline{x}_s, \overline{u}_s)} d\Pb(\overline{u}_{T-1} \mid \overline{a}_{T-1}, \overline{x}_T) - g(\overline{a}_T; \beta) \right. \\
	& \hphantom{= \int \frac{h(\overline{a}_T) \pi(\overline{a}_T)}{\prod_{s = 1}^{T-1}\pi(a_s \mid \overline{x}_s,  \overline{a}_{s-1})} \left\{\right.} \left. \vphantom{\int}  \quad - g(\overline{a}_T; \beta) \right\} p( \overline{a}_{T}, \overline{x}_T) d\overline{a}_Td \overline{x}_T \\
	& =  \int \frac{h(\overline{a}_T) \pi(\overline{a}_T)}{\prod_{s = 1}^{T-1}\pi(a_s \mid \overline{x}_s,  \overline{a}_{s-1})} \left\{  \int \E(Y^{\overline{a}_T} \mid \overline{a}_{T-1}, \overline{x}_T, \overline{u}_{T-1}) \frac{\prod_{s = 1}^{T-1} \pi(a_s \mid \overline{a}_{s-1} \overline{x}_s)}{\prod_{s = 1}^{T-1} \pi(a_s \mid \overline{a}_{s-1}, \overline{x}_s, \overline{u}_s)} d\Pb(\overline{u}_{T-1} \mid \overline{a}_{T-1}, \overline{x}_T) \right. \\
	& \hphantom{= \int \frac{h(\overline{a}_T) \pi(\overline{a}_T)}{\prod_{s = 1}^{T-1}\pi(a_s \mid \overline{x}_s,  \overline{a}_{s-1})} \left\{\right.} \left. \vphantom{\int}  \quad - g(\overline{a}_T; \beta) \right\} p( \overline{a}_{T-1}, \overline{x}_T) d\overline{a}_Td \overline{x}_T \\
	& =  \int \frac{h(\overline{a}_T) \pi(\overline{a}_T)}{\prod_{s = 1}^{T-1}\pi(a_s \mid \overline{x}_s,  \overline{a}_{s-1})} \left\{  \int \E(Y^{\overline{a}_T} \mid \overline{a}_{T-1}, \overline{x}_{T-1}, \overline{u}_{T-1}) \frac{\prod_{s = 1}^{T-1} \pi(a_s \mid \overline{a}_{s-1} \overline{x}_s)}{\prod_{s = 1}^{T-1} \pi(a_s \mid \overline{a}_{s-1}, \overline{x}_s, \overline{u}_s)} d\Pb(\overline{u}_{T-1} \mid \overline{a}_{T-1}, \overline{x}_{T-1}) \right. \\
	& \hphantom{= \int \frac{h(\overline{a}_T) \pi(\overline{a}_T)}{\prod_{s = 1}^{T-1}\pi(a_s \mid \overline{x}_s,  \overline{a}_{s-1})} \left\{\right.} \left. \vphantom{\int}  \quad - g(\overline{a}_T; \beta) \right\} p( \overline{a}_{T-1}, \overline{x}_{T-1}) d\overline{a}_Td \overline{x}_{T-1} 
\end{align*}
Repeating this calculation $T-1$ times, we arrive at
\begin{align*}
	& = \int \frac{h(\overline{a}_T) \pi(\overline{a}_T)}{\pi(a_1 \mid x_1)} \left\{  \int \E(Y^{\overline{a}_T} \mid a_1,x_1, u_1) \frac{ \pi(a_1 \mid x_1)}{ \pi(a_1 \mid, x_1, u_1)} d\Pb(u_1\mid a_1, x_1) - g(\overline{a}_T; \beta) \right\} p(a_1, x_1) d\overline{a}_Td x_1 \\
	& =  \int h(\overline{a}_T) \pi(\overline{a}_T) \left\{  \int \E(Y^{\overline{a}_T} \mid x_1, u_1)d\Pb(u_1\mid x_1) - g(\overline{a}_T; \beta) \right\} p(x_1) d\overline{a}_Td x_1 \\
	& =  \int h(\overline{a}_T) \pi(\overline{a}_T) \left\{ \E(Y^{\overline{a}_T}) - g(\overline{a}_T; \beta) \right\} d\overline{a}_T = 0
\end{align*}

\subsection{Additional useful lemmas}

\begin{lemma}[Theorem 12.3 in \cite{van2000asymptotic}]\label{lemma:u_vdv}
Let $h(z_1, z_2)$ be a symmetric function of two variables and $\E\{h^2(Z_1, Z_2)\} < \infty$. Then, 
\begin{align*}
\sqrt{n} (\mathbb{U}_n - \mathbb{U})h(Z_1, Z_2) \indist N\left(0, 4 \var \{h_1(Z_1)\} \right)
\end{align*}
where $h_1(Z_1) = \int h(Z_1, z_2)d\Pb(z_2)$.
\end{lemma}

\begin{lemma}[Rudelson LLN for Matrices, Lemma 6.2 in \cite{belloni2015some}] \label{lemma:rudelson}
Let $Q_1, \ldots, Q_n$ be a sequence of independent
symmetric, nonnegative $k \times k$-matrix valued 
random variables with $k \geq 2$ such that $Q = \Pn\{\E(Q_i)\}$ and $\|Q_i\| \leq M$ a.s.. Then, for $\widehat{Q} = \Pn Q$:
\begin{align*}
\E\|\widehat{Q} - Q\| \lesssim \frac{M \log k}{n} + \sqrt{\frac{M \|Q\| \log k}{n}}.
\end{align*}
\end{lemma}

\begin{lemma}\label{lemma:edward_u}
Let $\widehat{h}(z_1, z_2)$ be a symmetric function estimated on a separate training sample $D^n$ and
\begin{align*}
\widehat\Delta(z_1) = \int\{ \widehat{h}(z_1, z_2) - h(z_1, z_2) \}d\Pb(z_2).
\end{align*}
If $\E\left[\left\{\widehat{h}(Z_1, Z_2) - h(Z_1, Z_2) \right\}^2  |  D^n\right] < \infty$, then 
\begin{align*}
(\mathbb{U}_n - \mathbb{U}) \left\{\widehat{h}(Z_1, Z_2) - h(Z_1, Z_2)\right\} = O_\Pb\left( \frac{ \|\widehat\Delta\| }{\sqrt{n}} \right).
\end{align*}
\end{lemma}

\begin{proof}
We have
\begin{align*}
\E\left[(\mathbb{U}_n - \mathbb{U}) \left\{\widehat{h}(Z_1, Z_2) - h(Z_1, Z_2)\right\}  |  D^n \right] = 0
\end{align*}
because U-statistics are unbiased and $\widehat{h}(z_1, z_2)$ is a
fixed function given $D^n$. Let $\theta = \int f(z_1, z_2) d\Pb(z_1)d\Pb(z_2)$. The variance of a
U-statistic with symmetric kernel $f$ satisying $\E f^2(Z_1, Z_2) <\infty$ is
\begin{align*}
\var\{\mathbb{U}_n f(Z_1, Z_2) \} & = {n \choose 2}^{-2} \sum_{ 1 \leq i < j \leq n}  \sum_{ 1 \leq k < l \leq n} \int \{f(z_i, z_j)- \theta\}\{f(z_k, z_l) - \theta\} d\Pb(z_i) d\Pb(z_j) d\Pb(z_k) d\Pb(z_l) \\
& =  {n \choose 2}^{-2}  {n \choose 2} \cdot 2 \cdot (n-1) 
\var\left\{\int f(Z_1, z_2) d\Pb(z_2) \right\} + {n \choose 2}^{-2}  {n \choose 2} \var\{f(Z_1, Z_2) \} \\
& = \frac{4}{n} \var\left\{ \int f(Z_1, z_2) d\Pb(z_2) \right\} + o(n^{-1}) \\
& \leq \frac{4}{n} \E\left[\left\{ \int f(Z_1, z_2) d\Pb(z_2) \right\}^2\right] + o(n^{-1}).
\end{align*}
Substituting $f(z_1, z_2) = \widehat{h}(z_1, z_2) - h(z_1, z_2)$ into the expression above, we get
\begin{align*}
\var\left[(\mathbb{U}_n - \mathbb{U}) \{\widehat{h}(Z_1, Z_2) - h(Z_1, Z_2)\}  |  D^n \right] \leq  \frac{4 \|\widehat\Delta\|^2}{n} + o(n^{-1}) .
\end{align*}
The result then follows from Chebyshev's inequality. 
\end{proof}

\begin{lemma}\label{lemma:sjfacts}
	For $j = \{\ell, u\}$, let $s_j(Z; q_j) = q_j(Y  |  A, X) + \{Y - q_j(A, X)\} c_j^{\sgn\{Y - q_j(Y  |  A, X)\}}$, where $c_\ell = \gamma^{-1}$, $c_u = \gamma$, $q_\ell(Y  |  A, X)$ is the $1/(1 + \gamma)$-quantile of $Y$ given $(A, X)$, $q_u(Y  |  A, X)$ is the $\gamma / (1 + \gamma)$-quantile of $Y$ given $(A, X)$ and $\kappa(A, X; q_j) = \E\{s(Z; q_j)  |  A, X\}$. Then, the following holds:
	\begin{enumerate}
		\item The map $q \mapsto s(Z; q)$ is Lipschitz;
		\item The first and second derivatives of $q \mapsto \kappa(a, x; q)$ are
		\begin{align*}
			& \frac{d}{d q} \kappa(A,X; q) = 1 - c_j^{-1}\int_{-\infty}^{q} f(y  |  A =a, X = x) dy - c_j \int _{q}^\infty f(y  |  A =a, X = x) dy \\
			& \frac{d^2}{d q^2} \kappa(A,X; q)  = - c_j^{-1} f(q  |  A = a, X = x) + c_j f(q  |  A = a, X = x);
		\end{align*}
		\item The first derivative of $q \mapsto \kappa(a, x; q)$ vanishes at the true quantile $q_j(Y  |  A = a, X = x)$. 
	\end{enumerate}
\end{lemma}
\begin{proof}
	All three statements were noted by \cite{dorn2021doubly}. To prove the first one, let $q_1 < q_2$ without loss of generality and notice that if either $y < q_1 < q_2$ or $q_1 < q_2 < y$, 
	\begin{align*}
		\left| s(Z; q_1) - s(Z; q_2) \right| = \left| q_1 - q_2 + (q_2 - q_1) c_j^{\sgn\{y - q_1\}} \right| \leq (1 + \gamma) | q_1 - q_2|
	\end{align*}
because $y - q_1$ and $y - q_2$ agree on the sign. If $q_1 < y < q_2$,  $|y - q_1| \leq |q_1 - q_2|$ and $|y - q_2| \leq |q_1 - q_2|$ so that
\begin{align*}
	\left| s(Z; q_1) - s(Z; q_2) \right| = \left| q_1 - q_2 + (y - q_1)c_j^{\sgn\{y - q_1\}} + (y - q_2)c_j^{\sgn\{y - q_2\}} \right| \leq (1 + \gamma^{-1} + \gamma) | q_1 - q_2|
\end{align*}
The second statement follows from an application of Leibniz rule of integration and the third by noticing that
\begin{align*}
& \frac{d}{d q} \kappa(A,X; q) \Big|_{q = q_\ell} = 1 - \gamma \cdot \frac{1}{1 + \gamma} - \gamma^{-1} \left(1 -  \frac{1}{1 + \gamma} \right) = 0 \\
& \frac{d}{d q} \kappa(A,X; q)\Big|_{q = q_u} = 1 - \gamma^{-1} \cdot \frac{\gamma}{1 + \gamma} - \gamma \left(1 -  \frac{\gamma}{1 + \gamma} \right) = 0.
\end{align*}
\end{proof}

\begin{lemma}\label{lemma:margin}
Let $f(A)$ be a fixed function of the random variable $A$ with density upper bounded by $B$ and $g(A)$ be any other fixed function. Then, 
\begin{align*}
\left| \int \left[ \one\{g(a) \leq 0\} - \one\{f(a) \leq 0\} f(a) d\Pb(a)\right] \right| \leq 2B \|f - g\|^2_\infty.
\end{align*}
\end{lemma}

\begin{proof}
By Lemma 1 in \cite{kennedy2020sharp}, 
\begin{align*}
\left| \one\{g(a) \leq 0\} - \one\{f(a) \leq 0\} \right| \leq \one \{|f(a)| \leq | f(a) - g(a)| \} \leq \one \{|f(a)| \leq \| f - g\|_\infty \}.
\end{align*}
Therefore,
\begin{align*}
	\left| \int \left[ \one\{g(a) \leq 0\} - \one\{f(a) \leq 0\} f(a) d\Pb(a)\right] \right| & \leq \int  \one \{|f(a)| \leq \|f - g\|_\infty \} |f(a)| d\Pb(a) \\
	& \leq \| f - g\|_\infty \int \one \{|f(a)| \leq \|f - g\|_\infty \} d\Pb(a) \\
	& = \| f - g\|_\infty \Pb \left( -\|f - g\|_\infty \leq f(A) \leq \|f - g\|_\infty\right) \\
	& \leq 2B \| f - g\|^2_\infty.
\end{align*}
\end{proof}

\begin{lemma}\label{lemma:z_estimation}
	Let $\widehat{f}(z_1, z_2)$ be a function estimated on a separate independent sample and $g(A; \beta)$ be some parametric model indexed by $\beta \in \mathcal{B} \subset \R^k$. For some finite collection of known functions $h_1(A), \ldots, h_k(A)$, define
	\begin{align*}
		& \Psi_{n,l}(\beta) = \mathbb{U}_n\left[h_l(A_1)\left\{\widehat{f}(Z_1, Z_2) - g(A_1; \beta)\right\}\right] \text{ and } \Psi_l(\beta) =  \mathbb{U}\left[h_l(A_1)\left\{f(Z_1, Z_2) - g(A_1; \beta)\right\}\right].
	\end{align*}
and let $\Psi_n(\beta) = [\Psi_{n,1}(\beta), \ldots, \Psi_{n,k}(\beta)]$ and $\Psi(\beta)$ be defined similarly. Let $\widehat\beta_n$ and $\beta_0$ be the solutions to $\Psi_n(\widehat\beta) = o_\Pb(n^{-1/2})$ and $\Psi(\beta_0) = 0$, respectively, with $\beta_0$ in the interior of $\mathcal{B}$. Suppose that
\begin{enumerate}
	\item $\left\| \int S_2\left\{\widehat{f}(Z_1, z_2) - f(Z_1, z_2)\right\} d\Pb(z_2) \right\| = o_\Pb(1)$;
	\item The function class $\mathcal{G} = \left\{a \mapsto h_l(a)g(a; \beta), \beta \in \R^k \right\}$ is Donsker for every $l = \{1, \ldots, k\}$ with integrable envelop and $g(a; \beta)$ is a continuous function of $\beta$;
\item The function $\beta \mapsto \Psi(\beta)$ is differentiable at all $\beta$ with continuously invertible matrices 
$\dot{\Psi}_{\beta_0}$ and $\dot{\Psi}_{\widehat\beta}$, 
where $\dot{\Psi}_{\beta} = -\E\left\{h(A)\nabla^T_\beta g(A; \beta)\right\}$.
	\item $\max_l \left|\mathbb{U} \left[h_l(A_1) \left\{ \widehat{f}(Z_1, Z_2) - f(Z_1, Z_2)\right\} \right] \right| = o_\Pb(1)$. 
\end{enumerate}
Then, 
\begin{enumerate}
\item $\| \widehat\beta - \beta \| = o_\Pb(1)$;
\item $\widehat\beta - \beta =  \dot{\Psi}^{-1}_{\widehat\beta} \mathbb{U}h(A_1)\left\{\widehat{f}(Z_1, Z_2) - f(Z_1, Z_2) \right\} -  \dot{\Psi}^{-1}_{\beta_0} (\mathbb{U}_n - \mathbb{U}) h(A_1)\{f(Z_1, Z_2) - g(A_1; \beta)\} + o_\Pb(n^{-1/2}) $;
\item In particular, if  $\dot{\Psi}^{-1}_{\widehat\beta} \mathbb{U}\left[h(A_1)
\left\{\widehat{f}(Z_1, Z_2) - f(Z_1, Z_2) \right\}\right] = o_\Pb(n^{-1/2})$, then 
\begin{align*}
\sqrt{n}\left(\widehat\beta - \beta\right) \indist -\dot{\Psi}^{-1}_{\beta_0} N(0, 4\Sigma)
\end{align*}where
\begin{align*}
\Sigma = \E\left[ \int S_2h(A_1)\left\{f(Z_1, z_2) - g(A_1; \beta_0) \right\} d\Pb(z_2) \right]^2.
\end{align*}
\end{enumerate}
\end{lemma}

\begin{proof}
	Statement 1 follows from Theorem 2.10 in \cite{kosorok2008introduction}. We need to verify the two conditions of the theorem, namely:
	\begin{enumerate}
		\item $\| \Psi(\beta_n) \| \to 0$ implies $\|\beta_n - \beta_0\| \to 0$ for any sequence $\beta_n \in \R^k$;
		\item $\sup_{\beta \in \R^k} \|\Psi_n(\beta) - \Psi(\beta)\| = o_\Pb(1)$. 
	\end{enumerate}
By differentiability of $\Psi(\beta): \R^k \to \R^k$, 
\begin{align*}
	\Psi(\beta_n) = \dot{\Psi}(\beta_0)(\beta_n - \beta_0) + o(\|\beta_n - \beta_0\|) \implies \beta_n - \beta_0 + o(\|\beta_n - \beta_0\|) = \dot{\Psi}^{-1}(\beta_0)\Psi(\beta_n)
\end{align*}
Therefore, $\|\beta_n - \beta_0\|\{1+o(1)\} \lesssim \|\Psi(\beta_n)\| \to 0$. In addition, 
\begin{align*}
	\Psi_n(\beta) - \Psi(\beta) & = (\mathbb{U}_n - \mathbb{U}) \left[h(A_1)\left\{\widehat{f}(Z_1, Z_2) - f(Z_1, Z_2)\right\}\right] \\
	& \hphantom{=} + (\mathbb{U}_n - \mathbb{U}) h(A_1)f(Z_1, Z_2) + \mathbb{U} \left[h(A_1)\left\{\widehat{f}(Z_1, Z_2) - f(Z_1, Z_2)\right\} \right] \\
	& \hphantom{=} + (\Pn - \Pb) h(A) g(A; \beta) \\
	& = o_\Pb(1) + (\Pn - \Pb) h(A; \beta)g(A; \beta)
\end{align*}
Because a Donsker class is also Glivenko-Cantelli,  $\sup_{\beta \in \R^k} |(\Pn - \Pb) h_l(A)g(A; \beta)| = o_\Pb(1)$. Therefore, since $k$ is fixed,  $\sup_{\beta \in \R^k} \left\| (\Pn - \Pb) h(A) g(A; \beta)\right\| = o_\Pb(1)$. 

To prove Statement 2, we apply  Theorem 2.11 in \cite{kosorok2008introduction} to the ``debiased" moment condition
\begin{align*}
	\widetilde\Psi_n(\widetilde\beta) =  \mathbb{U}_n\left[h(A_1)\left\{\widehat{f}(Z_1, Z_2) - g(A_1; \widetilde\beta)\right\}\right] -  \mathbb{U}\left[h(A_1)\left\{\widehat{f}(Z_1, Z_2) - f(Z_1, Z_2)\right\}\right] = o_\Pb(n^{-1/2})
\end{align*}
By the same reasoning used to derive statement 1, we have that $\|\widetilde\beta - \beta_0\| = o_\Pb(1)$. Next, we have
\begin{align*}
		\sqrt{n}(\widetilde\Psi_n - \Psi)(\beta_0) & = \sqrt{n}(\mathbb{U}_n - \mathbb{U}) h(A_1)\{f(Z_1, Z_2) - g(A_1; \beta_0)\} \\& 
		\hphantom{=} + \sqrt{n} (\mathbb{U}_n - \mathbb{U})h(A_1)\left\{\widehat{f}(Z_1, Z_2) - f(Z_1, Z_2) \right\}
\end{align*}
Thus, by condition 1 and Lemma \ref{lemma:edward_u} together with Lemma \ref{lemma:u_vdv},  $\sqrt{n}(\widetilde\Psi_n - \Psi)(\beta_0) \indist N(0, 4\Sigma)$. 
Condition 2.12 in Theorem 2.11 in \cite{kosorok2008introduction} requires that
\begin{align*}
	\left\| \sqrt{n}(\Pn - \Pb)\left\{h(A)g(A; \widetilde\beta) - h(A)g(A; \beta_0) \right\} \right\| = o_\Pb\left(1 + \sqrt{n}\|\widetilde\beta - \beta_0\|\right)
\end{align*}
Because each function class $\mathcal{G}_l = \{a \mapsto h_l(a)g(a; \beta), \beta \in \R^k\}$ is Donsker, the process $\sqrt{n}(\Pn - \Pb) f$, $f \in \mathcal{G}$, is stochastically equicontinuous relative to the norm $\rho^2(f_1, f_2) = \var(f_1 - f_2) \leq \|f_1 - f_2\|^2$. In this respect, because $\|\widetilde\beta - \beta_0\| = o_\Pb(1)$ and $k$ is fixed, the condition above is satisfied. Therefore, we conclude that
\begin{align*}
	\widetilde\beta - \beta_0 = -\dot{\Psi}^{-1}_{\beta_0}(\mathbb{U}_n - \mathbb{U}) h(A_1)\{f(Z_1, Z_2) - g(A_1; \beta_0)\} + o_\Pb(n^{-1/2}).
\end{align*}
Finally, because $h_j(a)g(a; \beta)$ belongs to a Donsker class and $\| \widetilde\beta - \widehat\beta\| \leq \| \widetilde\beta - \beta_0\| + \| \widehat\beta - \beta_0\| = o_\Pb(1)$:
\begin{align*}
	o_\Pb(n^{-1/2}) & =\left\{ \Psi_n(\widehat\beta) - \widetilde\Psi_n(\widetilde\beta)\right\} \\
	& = (\Pn - \Pb) h(A)\left\{g(A; \widetilde\beta) - g(A; \widehat\beta)  \right\} +  \Pb h(A)\left\{g(A; \widetilde\beta) - g(A; \widehat\beta)  \right\} \\
	& \hphantom{=} +  \mathbb{U}\left[h(A_1)\left\{\widehat{f}(Z_1, Z_2) - f(Z_1, Z_2)\right\}\right] \\
	& = o_\Pb(n^{-1/2}) + o(\|\widetilde\beta - \widehat\beta\|) + \dot{\Psi}_{\widehat\beta}(\widetilde\beta - \widehat\beta) + \mathbb{U}\left[h(A_1)\left\{\widehat{f}(Z_1, Z_2) - f(Z_1, Z_2)\right\}\right]
\end{align*}
Rearranging, we have 
$$\widehat\beta - \widetilde\beta = \dot{\Psi}^{-1}_{\widehat\beta} \mathbb{U}\left[h(A_1)\left\{\widehat{f}(Z_1, Z_2) - f(Z_1, Z_2)\right\}\right] + o_\Pb(n^{-1/2}).$$

This concludes our proof that
\begin{align*}
	\widehat\beta - \beta_0 & = -\dot{\Psi}^{-1}_{\beta_0}(\mathbb{U}_n - \mathbb{U}) h(A_1)\{f(Z_1, Z_2) - g(A_1; \beta_0)\} + \dot{\Psi}^{-1}_{\widehat\beta} \mathbb{U}\left[h(A_1)\left\{\widehat{f}(Z_1, Z_2) - f(Z_1, Z_2)\right\}\right] \\
	& \hphantom{=} + o_\Pb(n^{-1/2}).
\end{align*}
\end{proof}

\end{document}